\documentclass[a4paper,UKenglish,cleveref, autoref, thm-restate]{lipics-v2021}

\pdfoutput=1 

\newif\ifincludeappendix
\includeappendixtrue

\ifincludeappendix
  \newcommand{\appendixref}[1]{\cref{#1}}
\else
  \newcommand{\appendixref}[1]{the Appendix}
\fi

\title{Universality Frontier for Asynchronous Cellular Automata} 

\author{Ivan Baburin}{Department of Computer Science,
  ETH Z{\"u}rich, Switzerland}{ivan.baburin@alumni.ethz.ch}{}{}
  
\author{Matthew Cook}{Faculty of Science and Engineering,
  University of Groningen}{m.cook@rug.nl}{}{}

\author{Florian Grötschla}{Department of Electrical Engineering,
  ETH Z{\"u}rich, Switzerland}{fgroetschla@ethz.ch}{}{}

\author{Andreas Plesner}{Department of Electrical Engineering,
  ETH Z{\"u}rich, Switzerland}{aplesner@ethz.ch}{}{}

\author{Roger Wattenhofer}{Department of Electrical Engineering,
  ETH Z{\"u}rich, Switzerland}{wattenhofer@ethz.ch}{}{}

\authorrunning{I. Baburin et al.} 

\Copyright{Ivan Baburin, Matthew Cook, Florian Grötschla, Andreas Plesner, Roger Wattenhofer} 



\relatedversiondetails[cite={baburin2025universality}]{Full version with appendix}{https://arxiv.org/abs/2502.05989} 

\nolinenumbers 

\EventEditors{Pawe\l{} Gawrychowski, Filip Mazowiecki, and Micha\l{} Skrzypczak}
\EventNoEds{3}
\EventLongTitle{50th International Symposium on Mathematical Foundations of Computer Science (MFCS 2025)}
\EventShortTitle{MFCS 2025}
\EventAcronym{MFCS}
\EventYear{2025}
\EventDate{August 25--29, 2025}
\EventLocation{Warsaw, Poland}
\EventLogo{}
\SeriesVolume{345}
\ArticleNo{78}

\usepackage[noend]{algpseudocode} 
\usepackage{circuitikz}
\usepackage[utf8]{inputenc}
\usepackage[T1]{fontenc}
\usepackage{tikz-cd}
\usepackage{tikz}
\usepackage{adjustbox}
\usepackage{subcaption}
\usepackage{cleveref}
\usepackage{epigraph}
\usepackage{float}
\usepackage{caption}
\usepackage{subcaption}
\usepackage{cancel}
\usepackage{svg}
\usepackage{csvsimple}
\usepackage{mathtools}
\usepgflibrary{patterns} 
\usetikzlibrary{patterns.meta}
\usepackage{seqsplit}
\usetikzlibrary{matrix, positioning, chains,fit,shapes, math, calc, petri,      automata, fit, positioning, matrix, arrows.meta, arrows, decorations.pathmorphing, shapes, shapes.multipart, chains, shadows.blur, decorations.pathreplacing}

\DeclareMathOperator{\Ima}{Im}
\definecolor{darkred}{rgb}{0.8,0,0}

\newcommand{\N}{\mathbb{N}}
\newcommand{\Z}{\mathbb{Z}}
\newcommand{\PS}[1]{\mathcal{P}(#1)}
\newcommand{\rle}[1]{\textsf{Rule #1}}
\newcommand{\pf}{\text{pf}}
\newcommand{\spacetime}{\text{st}}

\newcommand{\gol}{\textsf{Game of Life}}

\newcommand{\radiuszero}{
    \begin{tikzpicture}
        \tikzset{
          square matrix/.style={
            matrix of nodes,
            column sep=-\pgflinewidth, row sep=-\pgflinewidth,
            nodes={
              rectangle,
              draw=black,
              anchor=center,
              align=center,    
              inner sep=1pt
            },
            b/.style={draw, fill=black, minimum size=0.5em},
            w/.style={draw,  minimum size=0.5em},
            g/.style={draw, fill=gray, minimum size=0.5em}
          },
          square matrix/.default=1em
        }
        
        \matrix[square matrix, ampersand replacement=\&]
        {
        |[w]| \& |[b]| \& |[w]| \\
        };
    \end{tikzpicture}
}

\newcommand{\oneway}{
    \begin{tikzpicture}
        \tikzset{
          square matrix/.style={
            matrix of nodes,
            column sep=-\pgflinewidth, row sep=-\pgflinewidth,
            nodes={
              rectangle,
              draw=black,
              anchor=center,
              align=center,    
              inner sep=1pt
            },
            b/.style={draw, fill=black, minimum size=0.5em},
            w/.style={draw,  minimum size=0.5em},
            g/.style={draw, fill=gray, minimum size=0.5em}
          },
          square matrix/.default=1em
        }
        
        \matrix[square matrix, ampersand replacement=\&]
        {
        |[g]| \& |[b]| \& |[w]| \\
        };
    \end{tikzpicture}
}

\newcommand{\radiusone}{
    \begin{tikzpicture}
        \tikzset{
          square matrix/.style={
            matrix of nodes,
            column sep=-\pgflinewidth, row sep=-\pgflinewidth,
            nodes={
              rectangle,
              draw=black,
              anchor=center,
              align=center,    
              inner sep=1pt
            },
            b/.style={draw, fill=black, minimum size=0.5em},
            w/.style={draw,  minimum size=0.5em},
            g/.style={draw, fill=gray, minimum size=0.5em}
          },
          square matrix/.default=1em
        }
        
        \matrix[square matrix, ampersand replacement=\&]
        {
        |[g]| \& |[b]| \& |[g]| \\
        };
    \end{tikzpicture}
}

\newcommand{\vonNeumann}{
    \begin{tikzpicture}
        \tikzset{
          square matrix/.style={
            matrix of nodes,
            column sep=-\pgflinewidth, row sep=-\pgflinewidth,
            nodes={
              rectangle,
              draw=black,
              anchor=center,
              align=center,    
              inner sep=1pt
            },
            b/.style={draw, fill=black, minimum size=0.5em},
            w/.style={draw,  minimum size=0.5em},
            g/.style={draw, fill=gray, minimum size=0.5em}
          },
          square matrix/.default=1em
        }
        
        \matrix[square matrix, ampersand replacement=\&]
        {
        |[w]| \& |[g]| \& |[w]| \\
        |[g]| \& |[b]| \& |[g]| \\
        |[w]| \& |[g]| \& |[w]| \\
        };
    \end{tikzpicture}
}

\newcommand{\Custom}[5]{
    \begin{tikzpicture}
        \tikzset{
          square matrix/.style={
            matrix of nodes,
            column sep=-\pgflinewidth, row sep=-\pgflinewidth,
            nodes={
              rectangle,
              draw=black,
              anchor=center,
              align=center,    
              inner sep=1pt
            },
            b/.style={draw, fill=black, minimum size=0.5em},
            w/.style={draw,  minimum size=0.5em},
            g/.style={draw, fill=gray, minimum size=0.5em}
          },
          square matrix/.default=1em
        }
        
        \matrix[square matrix, ampersand replacement=\&]
        {
         \& |[#1]| \&  \\
        |[#2]| \& |[#3]| \& |[#4]| \\
        \& |[#5]| \&  \\
        };
    \end{tikzpicture}
}

\newcommand{\CustomSmall}[1]{
    \begin{tikzpicture}
        \tikzset{
          square matrix/.style={
            matrix of nodes,
            column sep=-\pgflinewidth, row sep=-\pgflinewidth,
            nodes={
              rectangle,
              draw=black,
              anchor=center,
              align=center,    
              inner sep=1pt
            },
            b/.style={draw, fill=black, minimum size=0.5em},
            w/.style={draw,  minimum size=0.5em},
            g/.style={draw, fill=gray, minimum size=0.5em}
          },
          square matrix/.default=1em
        }
        
        \matrix[square matrix, ampersand replacement=\&]
        {
         |[#1]| \\
        };
    \end{tikzpicture}
}

\newcommand{\Moore}{
    \begin{tikzpicture}
        \tikzset{
          square matrix/.style={
            matrix of nodes,
            column sep=-\pgflinewidth, row sep=-\pgflinewidth,
            nodes={
              rectangle,
              draw=black,
              anchor=center,
              align=center,    
              inner sep=1pt
            },
            b/.style={draw, fill=black, minimum size=0.5em},
            w/.style={draw,  minimum size=0.5em},
            g/.style={draw, fill=gray, minimum size=0.5em}
          },
          square matrix/.default=1em
        }
        
        \matrix[square matrix, ampersand replacement=\&]
        {
        |[g]| \& |[g]| \& |[g]| \\
        |[g]| \& |[b]| \& |[g]| \\
        |[g]| \& |[g]| \& |[g]| \\
        };
    \end{tikzpicture}
}

\newcommand{\innerradiustwo}{
    \begin{tikzpicture}
        \tikzset{
          square matrix/.style={
            matrix of nodes,
            column sep=-\pgflinewidth, row sep=-\pgflinewidth,
            nodes={
              rectangle,
              draw=black,
              anchor=center,
              align=center,    
              inner sep=1pt
            },
            b/.style={draw, fill=black, minimum size=0.35em},
            w/.style={draw,  minimum size=0.35em},
            g/.style={draw, fill=gray, minimum size=0.35em}
          },
          square matrix/.default=1em
        }
        
        \matrix[square matrix, ampersand replacement=\&]
        {
        |[g]| \& |[g]| \& |[g]| \& |[g]| \& |[g]|\\
        |[g]| \& |[g]| \& |[g]| \& |[g]| \& |[g]|\\
        |[g]| \& |[g]| \& |[b]| \& |[g]| \& |[g]| \\
        |[g]| \& |[g]| \& |[g]| \& |[g]| \& |[g]|\\
        |[g]| \& |[g]| \& |[g]| \& |[g]| \& |[g]|\\
        };
    \end{tikzpicture}
}

\crefname{section}{Section}{Sections}
\crefname{table}{Table}{Tables}
\crefname{figure}{Figure}{Figures}
\crefname{algorithm}{Algorithm}{Algorithms}
\crefname{equation}{Eq.}{Eqs.}
\crefname{example}{Example}{Examples}
\crefname{fact}{Fact}{Facts}
\crefname{appendix}{Appendix}{Appendices}
\crefname{theorem}{Theorem}{Theorems}
\crefname{reTheorem}{Theorem}{Theorems}
\crefname{aquestion}{Question}{Questions}
\crefname{assumption}{Assumption}{Assumptions}
\crefname{lemma}{Lemma}{Lemmas}
\crefname{reLemma}{Lemma}{Lemmas}
\crefname{proposition}{Proposition}{Propositions}
\crefname{chapter}{Chapter}{Chapters}
\crefname{line}{line}{lines}
\crefname{principle}{Principle}{Principles}
\crefname{definition}{Definition}{Definitions}
\crefname{corollary}{Corollary}{Corollaries}
\crefname{Exercise}{Exercise}{Exercises}
\crefname{observation}{Observation}{Observations}

\newtheorem{openquestion}{Open question}[]

\algrenewcommand\algorithmicindent{1.0em}%
\newcommand{\rightcomment}[1]{{\color{gray} \(\triangleright\){\footnotesize\textit{#1}}}}
\algrenewcommand{\algorithmiccomment}[1]{\hfill \rightcomment{#1}}  
\algnewcommand{\LinesComment}[1]{\State {\color{black!50!green}\rightcomment{\parbox[t]{.95\linewidth-\leftmargin-\widthof{\(\triangleright\) }}{#1}}}}
\algnewcommand{\LineComment}[1]{\State {\color{black!50!green}\smaller \(\triangleright\) \parbox[t]{\linewidth-\leftmargin-\widthof{\(\triangleright\) }}{\it #1}\smallskip}} 
\algnewcommand{\InlineComment}[1]{\hfill {\color{black!50!green}\(\triangleright\) {\scriptsize \it #1}}}
\algrenewcommand\algorithmicindent{1.0em}%
\algrenewcommand\alglinenumber[1]{{\tiny\color{black!50}#1.}\hspace{-2pt}}
\newcommand{\algorithmicfunc}[1]{\textbf{def} {#1}:}
\algdef{SE}[FUNC]{Func}{EndFunc}[1]{\algorithmicfunc{#1}}{}
\makeatletter
\ifthenelse{\equal{\ALG@noend}{t}}%
  {\algtext*{EndFunc}}
  {}%
\makeatother

\begin{document}

\maketitle

\begin{abstract}
    In this work, we investigate computational aspects of asynchronous cellular automata (ACAs), a modification of cellular automata in which cells update independently, following an asynchronous update schedule. We introduce flip automata networks (FANs), a simple modification of automata networks that remain robust under any asynchronous updating order. By using FANs as a middleman, we show that asynchronous automata can efficiently simulate their synchronous counterparts with a linear memory overhead, which improves upon the previously established quadratic bound. Additionally, we address the universality gap for (a)synchronous cellular automata--the boundary separating universal and non-universal automata, which is still not fully understood. We tighten this boundary by proving that all one-way asynchronous automata lack universal computational power. Conversely, we establish the existence of a universal asynchronous $6$-state first-neighbor automaton in one dimension and a $3$-state von Neumann automaton in two dimensions, which represent the smallest known universal constructions to date.
\end{abstract}

\section{Introduction}

Perhaps the most famous example of a complex system is a cellular automaton (CA): a uniform array of infinitely many simple machines--cells--all behaving in the same manner and interacting using basic, entirely local rules. Somewhat paradoxically, the simplicity of local components often results in the formation of sophisticated global patterns, highlighting the elementary non-continuity underlying the very concept of complexity. A cellular automaton can be seen as complex if any given computation can be embedded into its evolution, a phenomenon known as computational universality. During the last decades, many astonishingly simple universal cellular automata were constructed \cite{Ollinger-chapter}, which, as a byproduct, resulted in the discovery of many novel universal models of computation \cite{rule-110, billiard-ball, collisions}. On the other hand, cellular automata were also used as abstractions for biological and physical models of complexity: \gol{} is a famous example of an artificial biological system \cite{gol} while Wolfram \cite{Wolfram2002} used CAs to model physical laws using local discrete processes.

Despite their mathematical beauty, cellular automata have a severe limitation when it comes to modeling the real world--the requirement of global synchronization. To mitigate this issue, asynchronous cellular automata (ACAs) have been proposed, where every cell performs interactions at its local speed, without the need for global synchronization. Although at first glance it may seem that asynchronous automata are very different from their synchronous peers, in reality, they share many similarities. Nakamura \cite{nakamura} has shown that for any synchronous CA, we can construct an ACA which, in a very natural way, simulates the synchronous CA. However, this simulation comes at the price of a quadratic memory overhead, with further optimizations introduced later \cite{improved-soldiers}. A unified, algebraic approach to asynchrony was later given by Gács \cite{gacs}, who characterized a large class of cellular automata whose dynamics remains invariant with respect to any choice of updating.

Asynchronous updating allows more realistic modeling of many natural phenomena, and many potential use cases in distributed computing have been proposed \cite{aca-survey}. For example, one can draw parallels to emergent swarm intelligence in biological systems, such as ant colonies, which use chaotic cooperation of many simple agents to solve increasingly complex transport problems \cite{ants-collective, ants-transport}. The implementation of asynchrony can also directly affect computational ability--for example, in chemical reaction networks, depending on the asynchronous setting, the computational power varies between primitive recursive functions, Boolean circuits, and Turing machines \cite{CRN}. Similarly, although cellular automata are capable of universal computation, a certain threshold in terms of their size and state count must be met to enable it. This question has been extensively investigated for CAs \cite{ollinger-4-states, lindgren, rule-110, durand-gol, lifewithoutdeath, Banks, billiard-ball, morita-book, Wolfram2002} and ACAs \cite{3-states-random, lee-5, lee, schneider, adachi}. Despite this, the limits of asynchronous universality are not yet known, and the smallest constructions still lag behind their synchronous counterparts, hinting at a foundational discrepancy between synchronous and asynchronous computation.

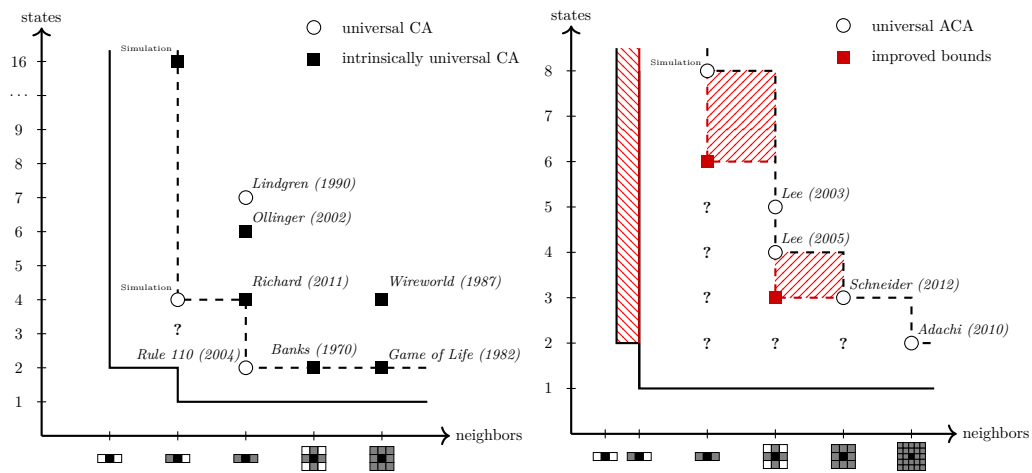
\begin{figure}[tb]
    \centering
    \begin{subfigure}{0.48\textwidth}
        \begin{tikzpicture}[scale=0.6, transform shape]
            \draw[->, thick] (0,0) -- (9,0) node[right] {neighbors};
            \draw[->, thick] (0,0) -- (0,9) node[above] {states};
            
            \node at (1.5,-0.5){
                \radiuszero
            };
            \node at (3,-0.5){
                \oneway
            };
            \node at (4.5,-0.5){
                \radiusone
            };
            \node at (6,-0.5){
                \vonNeumann
            };
            \node at (7.5,-0.5){
                \Moore
            };
            
            \foreach \i in {1,...,5} \draw (1.5*\i,-0.1)--(1.5*\i,0.1);
            \foreach \i in {1,...,9} \node at (-0.5, 0.75*\i){\i};
            \foreach \i in {1,...,11} \draw (-0.1, 0.75*\i)--(0.1, 0.75*\i);
            \node at (-0.5, 0.75*11){16};
            \node at (-0.5, 0.75*10){$\ldots$};
            
            \draw[thick, dashed] (3, 8.5) -- (3, 3) -- (4.5, 3) -- (4.5, 1.5) -- (8.5, 1.5);
            
            \draw[thick] (1.5, 8.5) -- (1.5,1.5) -- (3,1.5) -- (3, 0.75) -- (8.5, 0.75);
            
            \draw[fill=white] (3,3) circle[radius=1.5mm]; 
            \node[above left] at (3,3.1) {\tiny{Simulation}};
            \draw[fill=white] (4.5, 5.25) circle[radius=1.5mm]; 
            \node[above right] at (4.5,5.25) {\small{\textit{Lindgren (1990)}}};
            \draw[fill=white] (4.5,1.5) circle[radius=1.5mm]; 
            \node[above left] at (4.5,1.5) {\small{\textit{Rule 110 (2004)}}};
            \node[] at (3,2.35) {\textbf{?}};
            
            \filldraw (3, 8.25) node[rectangle, fill=black, minimum size = 1.5mm, inner sep = 4pt] {}; 
            \node[above left] at (3,8.35) {\tiny{Simulation}};
            \filldraw (4.5, 4.5) node[rectangle, fill=black, minimum size = 1.5mm, inner sep = 4pt] {}; 
            \node[above right] at (4.5,4.5) {\small{\textit{Ollinger (2002)}}};
            \filldraw (4.5, 3) node[rectangle, fill=black, minimum size = 1.5mm, inner sep = 4pt] {}; 
            \node[above right] at (4.5,3.1) {\small{\textit{Richard (2011)}}};
            \filldraw (6, 1.5) node[rectangle, fill=black, minimum size = 1.5mm, inner sep = 4pt] {}; 
            \node[above] at (6,1.6) {\small{\textit{Banks (1970)}}};
            \filldraw (7.5, 1.5) node[rectangle, fill=black, minimum size = 1.5mm, inner sep = 4pt] {}; 
            \node[above right] at (7.5,1.5) {\small{\textit{Game of Life (1982)}}};
            \filldraw (7.5, 3) node[rectangle, fill=black, minimum size = 1.5mm, inner sep = 4pt] {}; 
            \node[above right] at (7.5,3.1) {\small{\textit{Wireworld (1987)}}};

            \begin{scope}[shift={(6,9)}]
                \draw (0,0) circle[fill, radius=1.5mm] {};
                \node[right] at (0.5,0) {universal CA};

                \draw (0,-0.7) node[rectangle, fill=black, minimum size = 1.5mm, inner sep = 4pt] {};
                \node[right] at (0.5,-0.7) {intrinsically universal CA};
            
            \end{scope}
        \end{tikzpicture}
    \end{subfigure}
    \hspace{0.1em}
    \begin{subfigure}{0.48\textwidth}
        \centering
        \begin{tikzpicture}[scale=0.6, transform shape]
            \draw[draw=none, pattern=north east lines, pattern color=red] (1.5,8)--(3, 8)--(3, 6)--(1.5, 6)--(1.5, 8);
            \draw[draw=none, pattern=north east lines, pattern color=red] (3,4)--(4.5, 4)--(4.5, 3)--(3, 3)--(3, 4);
            \draw[draw=none, pattern=north west lines, pattern color=red] (-0.5, 2)--(-0.5, 8.5)--(0,8.5)--(0,2)--(-0.5,2);
            
            \draw[->, thick] (-1.5,0) -- (7,0) node[right] {neighbors};
            \draw[->, thick] (-1.5,0) -- (-1.5,9) node[above] {states};
            
            \node at (-0.75,-0.5){
                \radiuszero
            };
            \node at (0,-0.5){
                \oneway
            };
            \node at (1.5,-0.5){
                \radiusone
            };
            \node at (3,-0.5){
                \vonNeumann
            };
            \node at (4.5,-0.5){
                \Moore
            };
            \node at (6,-0.5){
                \innerradiustwo
            };
            \foreach \i in {0,...,4} \draw (1.5*\i,-0.1)--(1.5*\i,0.1);
            \foreach \i in {1,...,8} \node at (-2, \i){\i};
            \foreach \i in {1,...,8} \draw (-1.6, \i)--(-1.4, \i);
            \draw (-0.75,-0.1)--(-0.75,0.1);

            
            \draw[thick] (6.5, 1) -- (0, 1) -- (0, 8);
            \draw[thick] (0, 2) -- (0,1);
            \draw[thick] (0, 2) -- (-0.5,2)--(-0.5, 8.5);
            \draw[thick, color=darkred] (0, 8.5) -- (0,2);
            
            \draw[thick, dashed, color=darkred] (1.5, 8) -- (1.5, 6) -- (3, 6);
            \draw[thick, dashed, color=darkred] (3, 4) -- (3, 3) -- (4.5, 3);
            \draw[thick, dashed] (1.5, 8.5) -- (1.5, 8) --(3, 8) -- (3, 6) -- (3, 5) -- (3, 4) --(4.5, 4) --(4.5, 3) -- (6, 3) -- (6, 2) -- (6.5, 2);

            \draw[fill=white] (1.5,8) circle[radius=1.5mm]; 
            \node[above left] at (1.5,8) {\tiny{Simulation}};
            \draw[fill=white] (4.5,3) circle[radius=1.5mm]; 
            \node[above right] at (4.5,3) {\small{\textit{Schneider (2012)}}};
            \draw[fill=white] (3, 4) circle[radius=1.5mm]; 
            \node[above right] at (3,4) {\small{\textit{Lee (2005)}}};
            \draw[fill=white] (3, 5) circle[radius=1.5mm]; 
            \node[above right] at (3,5) {\small{\textit{Lee (2003)}}};
            \draw[fill=white] (6,2) circle[radius=1.5mm]; 
            \node[above right] at (6,2) {\small{\textit{Adachi (2010)}}};
            \node[] at (4.5,2) {\textbf{?}};
            \node[] at (3,2) {\textbf{?}};
            \node[] at (1.5,2) {\textbf{?}};
            \node[] at (1.5,3) {\textbf{?}};
            \node[] at (1.5,4) {\textbf{?}};
            \node[] at (1.5,5) {\textbf{?}};
            \filldraw (3, 3) node[rectangle, fill=darkred, minimum size = 1.5mm, inner sep = 4pt] {}; 
            \filldraw (1.5, 6) node[rectangle, fill=darkred, minimum size = 1.5mm, inner sep = 4pt] {}; 

            \begin{scope}[shift={(4.5,9)}]
                \draw (0,0) circle[fill, radius=1.5mm] {};
                \node[right] at (0.5,0) {universal ACA};

                \draw (0,-0.7) node[rectangle, fill=darkred, minimum size = 1.5mm, inner sep = 4pt] {};
                \node[right] at (0.5,-0.7) {improved bounds};
            
            \end{scope}
        \end{tikzpicture}
    \end{subfigure}
    \caption{Current ``universality frontier'' for synchronous (left) and asynchronous cellular automata (right). The solid line indicates the impossibility proofs for universality, while the dashed line shows the currently smallest known universal constructions. Contributions of this work are marked in red: $6$-state first neighbors, $3$-state von Neumann, and non-universality of one-way automata.}
    \label{fig:asynchronous-frontier}
\end{figure}

\subparagraph*{Our contributions} We present a comprehensive overview of computational abilities of (a)synchronous cellular automata and propose a formal definition for the notion of simulation with ACAs, which, although often used implicitly, has been missing in the literature \cite{nakamura, improved-soldiers, worsch-intrinsically, golze}. We adopt the ``adversarial'' view on asynchrony, i.e. require the ACAs to be robust against \emph{any} asynchronous updating (other notions such as randomized and ``best-case'' asynchrony have also been studied \cite{aca-survey}). By combining various simulation techniques and known universality constructions, we assemble an explicit ``universality frontier'' for CAs and ACAs, i.e., the boundary between universal and non-universal computation, similarly to how it was already done for Turing machines \cite{small-utms} (see \Cref{fig:asynchronous-frontier}). The dashed line represents the best-known upper bound on universality constructions, while the filled line represents the best-known lower bound for non-universality constructions. We further present our own contributions to both asynchronous simulation and the ``universality frontier'':
\begin{itemize}
    \item In \cref{sec:flip-automata-networks}, we introduce a simple asynchronous model of computation--the flip automata network (FAN), which have invariant computations under any asynchronous setting.
    \item In \cref{sec:simulation} we show that by using FANs as a middleman, ACAs can simulate their synchronous counterparts with just $4q$ states (and $3q$ in the $1$-dimensional case), improving on the previous constructions which had quadratic overhead.
    \item In \cref{sec:non-universal} we extend the ``universality frontier'' by deriving a non-universality result for one-way asynchronous automata, demonstrating their computational equivalence to classical finite-state machines. This suggests that bidirectional communication is crucial to achieving universality in ACAs (which is not the case for classical CAs).
    \item In \cref{sec:small-aca} we present smaller universal ACAs than previously known: a universal $3$-state von Neumann automaton and a universal $6$-state first neighbors automaton.
\end{itemize}
\section{Background}\label{sec:background}

\begin{definition}[Cellular Automaton]\label{def:cellular-automaton}
    A $d$-dimensional cellular automaton (CA) is a tuple $(S, N, f, \Z^d)$ where $f\colon S^k \longrightarrow S$ is the local rule function, $S$ is the set of states, $N = (n_1, n_2, \ldots, n_k)$ with $n_i \in \Z^d$ is the (local) neighborhood vector (defines neighbor offsets for each cell) and $\mathbb{Z}^d$ is the cellular space.
\end{definition}
We call a neighborhood $N$ \emph{symmetric} if $-1 \cdot N$ is a permutation of $N$, and two cells $i, j \in \Z^d$ are \emph{neighboring} if they are in each other's neighborhoods. There are infinitely many choices for the neighborhood vector $N$, but four types of neighborhoods are considered canonical and have been researched the most (note, Moore and von Neumann generalize naturally to arbitrary dimension $d$, and first neighbors is the $1$-dimensional case for both):
\begin{center}
    \emph{One-way}\raisebox{-3pt}{\oneway} \quad \emph{First neighbors}\raisebox{-3pt}{\radiusone} \quad \emph{Von Neumann}\raisebox{-6pt}{\vonNeumann} \quad \emph{Moore}\raisebox{-6pt}{\Moore}
\end{center}

A \emph{configuration} of a $d$-dimensional CA is a function $c: \Z^d \longrightarrow S$, which assigns a state to every cell in the cellular space. At each timestamp $t$, every cell $i \in \Z^d$ synchronously applies the local rule function to its local neighborhood $N$, causing the complete configuration to transition to a new configuration $c_{t+1}$. For each cell $i$ the next state $c_{t+1}(i)$ is defined as
\begin{equation*}
    c_{t+1}(i) = f(c_t(i + n_0), \ldots, c_t(i + n_k))
\end{equation*}

We can characterize the transformation between two consecutive configurations of a CA by specifying the \emph{global transition function} $G: S^{\Z^d} \longrightarrow S^{\Z^d}$ with $G(c_k) = c_{k+1}$. 

\begin{definition}[Asynchronous Cellular Automaton]\label{def:aca}
    An asynchronous cellular automaton (ACA) $\mathbb{A} = (S, N, f, \mathbb{Z}^d)$ is a cellular automaton that chooses at each point in time $t$ an update set $D_t \subseteq \Z^d$ and applies the local rule function $f$ only to cells in $D_t$.
\end{definition}

\Cref{def:aca} presents the most common type of asynchronous updating, also known as \emph{purely} asynchronous. Other notions such as \emph{$\alpha$-asynchrony} (each cell updates independently with probability $\alpha$) have also been considered, for a detailed survey we refer to Fatès \cite{aca-survey}. 

\begin{definition}[Update Schedule]
    For an asynchronous cellular automaton $\mathbb{A} = (S, N, f, \Z^d)$ we define some asynchronous update schedule $\zeta \colon \N \longrightarrow \PS{\Z^d}$, a function specifying the subsets of cells updated at each point $t$ in time. $\zeta$ has the property that every $i \in S^{\Z^d}$ appears infinitely often in $\Ima{\zeta}$, or, in other words, every cell is updated infinitely often.
\end{definition}

For a purely asynchronous setting, $\zeta$ can be arbitrary: an ``adversary'' can choose any $\zeta$ to update $\mathbb{A}$. Given a $\zeta$ we identify the asynchronous global transition function $G^\zeta \colon S^{\Z^d} \longrightarrow S^{\Z^d}$ and corresponding \emph{asynchronous space-time} $\spacetime^\zeta \colon \N \times \Z^d \longrightarrow S$ mapping configurations to their evolved asynchronous counterparts; $\spacetime^\zeta(0, c_0) = c_0$ and $\spacetime^\zeta(t, c_0) = c_t^\zeta$ such that
\begin{equation*}
    c^\zeta_{t+1} = G^\zeta(c_t) = G_{\zeta(t)}(c_t) \quad \text{where} \quad G_D(c)(i) = 
    \left\{ \begin{array}{rc} 
    G(c)(i) & \mbox{if} \quad i \in D \\ 
    c(i) & \mbox{otherwise} 
    \end{array}\right.
\end{equation*}
Furthermore, we call a cell $i \in \mathbb{Z}^d$ \emph{active} (on configuration $c$) if and only if $G(c)(i) \neq c(i)$ (note that this is independent of the update schedule).
\begin{definition}[Update History]
    Given an asynchronous cellular automaton $\mathbb{A}$ under some update schedule $\zeta$, we define its update history as a sequence of configurations $h_0, h_1, \ldots$ where $h_0 = c_0$ and $h_t(i) = c_{t'_i}(i)$ with $t'_i$ as below if such $t'_i$ exists, and undefined otherwise.
    \begin{equation*}
         t'_i = \min \{a_t \in \N \ | \ \exists a_0 \ldots a_{t-1} \ \forall j < t \ (a_j < a_{j+1} \wedge c_{a_j}(i) \neq c_{a_{j+1}}(i)) \}
    \end{equation*} 
\end{definition}

Intuitively, the update history records the asynchronous space-time ``up to value equality.'' We only record ``true'' state updates to cell $i$, discarding the updates that did not change the state and the time these updates occurred (in particular, the history can be finite!). We say that an ACA $\mathbb{A}$ has \emph{invariant histories on configuration $c_0$} if all update schedules initialized with $c_0$ result in the same update history. Further, $\mathbb{A}$ has \emph{invariant histories} if all initial configurations have invariant histories. Such ACAs form a particularly important subclass of CAs, since they allow for reliable computations without any need for global synchronization. 

Gács \cite{gacs} has given an almost complete characterization of ACAs with invariant histories using only algebraic properties of monotonicity and commutativity.

\begin{definition}[Monotonicity]
    A global transition function $G$ is \textit{monotonic} if for any configuration $c$, active cell $i$ and an update set $A \subseteq \mathbb{Z}^d \setminus \{i\}$ it holds that
    \begin{equation*}
        G_{\{i\}}(G_A(c))(i) \neq c(i).
    \end{equation*}
\end{definition}

An update of any other cell cannot stop cell $i$ from updating, ensuring a ``lock-free'' flow.

\begin{definition}[Commutativity]\label{def:commutativity}
    A global transition function $G$ is \textit{commutative} if for any configuration $c$ and two disjoint update sets of active cells $A, B \subseteq \mathbb{Z}^d$ it holds
    \begin{equation*}
        G_A(G_B(c)) = G_{A \cup B}(c) = G_B(G_A(c)).
    \end{equation*}
\end{definition}

Commutativity implies that updates of simultaneously active cells can be performed in any order. We say that the CA is \emph{commutative} if its global transition function is commutative. 

\begin{theorem}[Invariance \cite{gacs}]\label{thm:commutativity}
    A cellular automaton $(S, N, f, \mathbb{Z}^d)$ is commutative if and only if it has invariant histories and its global transition function $G$ is monotonic.
\end{theorem}

Even though commutativity does not characterize all automata with invariant histories (this question was shown to be undecidable in the same paper), it characterizes a very natural subclass thereof, since monotonicity is, in general, desirable. Furthermore, commutativity can be verified algorithmically--due to the local structure of $G$ it is sufficient to check the conditions in \cref{def:commutativity} only for cases where $A$ and $B$ are singleton update sets, i.e. whether $G_{\{i\}}(G_{\{j\}}(c)) = G_{\{i, j\}}(c)$ holds for any cells $i, j \in \mathbb{Z}^d$ (see \appendixref{app:examples}).

A generalization of commutativity for automata networks--an extension of cellular automata allowing arbitrary spaces and non-uniform local rule functions--is known, but requires more complex machinery \cite{automata-networks}. We note that while the algebraic characterization provides a simple way of verifying for invariant histories in ACAs, it is generally hard to design commutative ACAs with desired properties without inducing a large state complexity overhead. These issues arise due to the lack of a general ``constructive'' understanding of such ACAs, something we attempt to mitigate in the sections below. 

\section{Flip Automata Networks}\label{sec:flip-automata-networks}

In this section, we present a simple asynchronous distributed model of computation called flip automata networks (FAN). While FANs are just a simple modification of automata networks, they have the advantage of having invariant histories by design, allowing us to build a simple bridge between CAs and ACAs.

\begin{definition}[Flip Automata Network]\label{def:flip-automata-network}
    A FAN is a triple $\mathcal{A}_\text{flip} = (T, S, {\{f_i\}}_{i\in I})$ where $T = (I, E)$ is a undirected graph with some indexing $v_i \in I$ edges $E \subseteq \{ \{v_i, v_j \} \ | \ i \neq j \wedge v_i, v_j \in I \}$ such that each edge has a direction indicator $\in {\{ \nearrow, \swarrow \}}$, neighborhood function
    \begin{equation*}
        N(v_i) = \{ v \in I \ | \ \{v, v_i\} \in E \} 
    \end{equation*} 
    which must be finite $\forall v_i \in I$, $S$ is a state set, and ${\{f_i\}}_{i \in I}$ is the set of all flip functions where
    \begin{equation*}
        f_i\colon S^{N(v_i) \cup \{v_i\}} \longrightarrow S \times {\{ \nearrow, \swarrow \}}^{N(v_i)}
    \end{equation*}
    is the flip function associated with vertex $v_i$ ($\nearrow$ ``points'' to the node with the larger index and $\swarrow$ to the node with the smaller index). A function $f_i$ is only applicable if the direction indicator on all edges between $N(v_i)$ and $v_i$ ``point'' towards $v_i$, and the application of $f_i$ can update the state of $v_i$ and ``flip'' the direction indicator on any of the edges from $N(v_i)$. 
\end{definition}

Similarly to classical automata networks, individual transition functions can be applied in some (arbitrary) order specified by a given update schedule $\zeta \colon \N \longrightarrow \PS{I}$, and we define configuration of a FAN by assigning each node a state and each edge a direction indicator. FANs have an intuitive interpretation--the arrows indicate directed information sharing and every node waits for all of its neighbors to share before performing a transition. A small example is provided in \cref{fig:causal-network}.

\begin{figure}[htb]
    \centering
    \begin{tikzpicture}[
            > = stealth, 
            shorten > = 1pt, 
            auto,
            node distance = 2cm, 
            semithick, 
            scale=0.75, transform shape
        ]
        \tikzstyle{every state}=[
            draw = black,
            thick,
            fill = white,
            minimum size = 4mm
        ]
        \node[state, fill=lightgray!30] (s) {$1$};
        \node[state] (s1) [above right=0.7cm and 0.7cm of s] {$1$};
        \node[state] (s2) [below right=0.7cm and 0.7cm of s] {$1$};
        \node[state] (s3) [above left=0.7cm and 0.7cm of s] {$1$};
        \node[state] (s4) [below left = 0.7cm and 0.7cm of s] {$1$};

        \path[->] (s1) edge node{} (s);
        \path[->] (s2) edge node{} (s);
        \path[->] (s3) edge node{} (s);
        \path[->] (s4) edge node{} (s);
        \path[->] (s2) edge node{} (s1);
        \path[->] (s3) edge node{} (s1);
        \path[->] (s2) edge node{} (s4);
        \path[->] (s3) edge node{} (s4);
    \end{tikzpicture}
    \hspace{3em}
    \begin{tikzpicture}[
            > = stealth, 
            shorten > = 1pt, 
            auto,
            node distance = 2cm, 
            semithick, 
            scale=0.75, transform shape
        ]

        \tikzstyle{every state}=[
            draw = black,
            thick,
            fill = white,
            minimum size = 4mm
        ]

        \node[state] (s) {$2$};
        \node[state, fill=lightgray!30] (s1) [above right=0.7cm and 0.7cm of s] {$1$};
        \node[state] (s2) [below right=0.7cm and 0.7cm of s] {$1$};
        \node[state] (s3) [above left=0.7cm and 0.7cm of s] {$1$};
        \node[state, fill=lightgray!30] (s4) [below left = 0.7cm and 0.7cm of s] {$1$};

        \path[->] (s) edge node{}  (s1);
        \path[->] (s) edge node{} (s2);
        \path[->] (s) edge node{} (s3);
        \path[->] (s) edge node{} (s4);
        \path[->] (s2) edge node{} (s1);
        \path[->] (s3) edge node{} (s1);
        \path[->] (s2) edge node{} (s4);
        \path[->] (s3) edge node{} (s4);

    \end{tikzpicture}
    \hspace{3em}
    \begin{tikzpicture}[
            > = stealth, 
            shorten > = 1pt, 
            auto,
            node distance = 2cm, 
            semithick, 
            scale=0.75, transform shape
        ]

        \tikzstyle{every state}=[
            draw = black,
            thick,
            fill = white,
            minimum size = 4mm
        ]

        \node[state] (s) {$2$};
        \node[state] (s1) [above right=0.7cm and 0.7cm of s] {$2$};
        \node[state, fill=lightgray!30] (s2) [below right=0.7cm and 0.7cm of s] {$1$};
        \node[state, fill=lightgray!30] (s3) [above left=0.7cm and 0.7cm of s] {$1$};
        \node[state] (s4) [below left = 0.7cm and 0.7cm of s] {$2$};

        \path[->] (s1) edge node{}  (s);
        \path[->] (s) edge node{} (s2);
        \path[->] (s) edge node{} (s3);
        \path[->] (s4) edge node{} (s);
        \path[->] (s1) edge node{} (s2);
        \path[->] (s1) edge node{} (s3);
        \path[->] (s4) edge node{} (s2);
        \path[->] (s4) edge node{} (s3);
    \end{tikzpicture}
    \caption{Two-step evolution of flip automata network $\mathcal{A}_\text{flip} = (T_{\text{flip}}, \{0, 1, 2\}, \{f_i\}_{i=0}^4)$ where each $f_i$ computes the sum of all neighboring node states mod $3$, and flips all incoming arrows (applicable transitions marked in gray). At each point in time we apply all applicable transitions.}
    \label{fig:causal-network}
\end{figure}
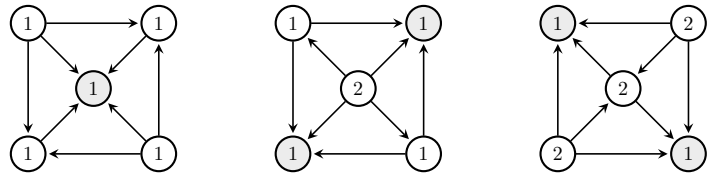

\begin{restatable}{proposition}{fliptheorem}\label{prop:flip-invariant}
    Flip automata networks have invariant histories, i.e., the update schedule does not affect the update history of individual nodes.
\end{restatable}

The proof for \Cref{prop:flip-invariant} can be found in \appendixref{app:ommited-proofs}. In the following sections, we demonstrate that flip automata networks can be embedded into asynchronous cellular automata with minimal overhead, enabling simpler simulation techniques and smaller universality constructions. A simple embedding is given in \Cref{ex:mountain-valley} below.

\begin{example}\label{ex:mountain-valley}
    Consider an asynchronous version of \rle{212}. The only transitions that affect the cell state are $(100) \mapsto 1$ and $(011) \mapsto 0$; in all other cases, the cell does not change its state. Observe that no two neighboring cells can be active at the same time. Using the local neighborhood of each cell, we can interpret \rle{212} as a $1$-dimensional flip automata network. To this end we define a mapping $\psi$ that adds direction indicators between neighboring cells $i$ and $i+1$ based on their current states in some configuration $c$:
    \begin{align*}
        &\psi(c(i), c(i+1)) = c(i) \longleftarrow c(i+1) \quad \text{for $c(i) = c(i+1)$}\\
        &\psi(c(i), c(i+1)) = c(i) \longrightarrow c(i+1) \quad \text{for $c(i) \neq c(i+1)$}
    \end{align*}
    This allows us to transform any configuration of \rle{212} into a configuration of a flip automata network. For example, the following synchronous update of configuration $c$
    \begin{center}
        $\ldots 100110011 \ldots \mapsto \ldots 110011001 \ldots$
    \end{center}
    can be seen as an update in the flip automata network, where all applicable transitions flip:
    \begin{center}
        $\fbox{1} \longrightarrow \fbox{0} \longleftarrow \fbox{0} \longrightarrow \fbox{1} \longleftarrow \fbox{1} \longrightarrow \fbox{0} \longleftarrow \fbox{0} \longrightarrow \fbox{1} \longleftarrow \fbox{1}$ \\
        $\Downarrow$\\
        $\fbox{1} \longleftarrow \fbox{1} \longrightarrow \fbox{0} \longleftarrow \fbox{0} \longrightarrow \fbox{1} \longleftarrow \fbox{1} \longrightarrow \fbox{0} \longleftarrow \fbox{0} \longrightarrow \fbox{1}$
    \end{center}
    Furthermore, if we think of $\longrightarrow$ as $\diagdown$ slope and of $\longleftarrow$ as $\diagup$ slope we can visualize any configuration of such flip automata network as a zigzag line in a $2$-dimensional space, which we will refer to as a \emph{mountain-valley} landscape, as shown in \Cref{fig:wave} (left). The update history of configuration $c$ under some asynchronous update schedule can thus be visualized as a $2$-dimensional ``fluid'' slowly filling the space-time $\spacetime$, where at each point in time some subset of ``valleys'' become ``mountains''. Depending on the update schedule $\zeta$, some mountains will grow faster than others; however, since every cell will be updated infinitely often, the ``fluid'' is guaranteed to fill the full space-time, as sketched in \Cref{fig:wave} (right). 
    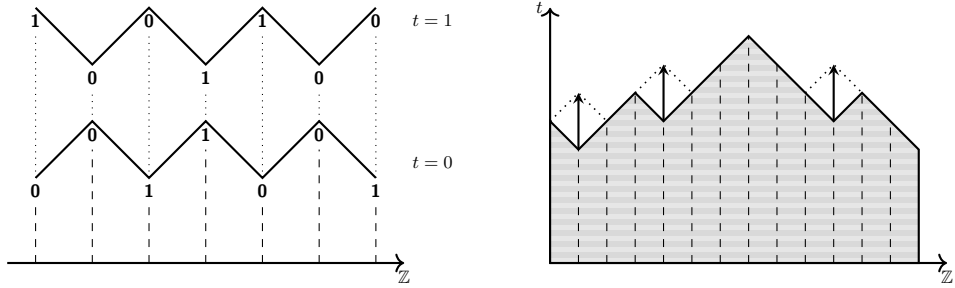
\begin{figure}[htb]
    \centering
    \begin{subfigure}{0.475\textwidth}
        \centering
        \begin{tikzpicture}[font=\sffamily\small, scale=0.75, transform shape]
        
        \draw[->, thick] (-0.5,0.5) -- (6.5,0.5) node[below] {$\mathbb{Z}$};

        \draw[black, thick] (0,2) -- (1, 3) -- (2, 2) -- (3, 3) -- (4, 2) -- (5, 3) -- (6, 2);
        \draw[black, thick] (0,5) -- (1, 4) -- (2, 5) -- (3, 4) -- (4, 5) -- (5, 4) -- (6, 5);

        \draw[dotted] (0, 2) -- (0, 4.5);
        \draw[dotted] (1, 3) -- (1, 3.5);
        \draw[dotted] (2, 2) -- (2, 4.5);
        \draw[dotted] (3, 3) -- (3, 3.5);
        \draw[dotted] (4, 2) -- (4, 4.5);
        \draw[dotted] (5, 3) -- (5, 3.5);
        \draw[dotted] (6, 2) -- (6, 4.5);

        \draw[dashed] (0, 0.5) -- (0, 1.5);
        \draw[dashed] (1, 0.5) -- (1, 2.5);
        \draw[dashed] (2, 0.5) -- (2, 1.5);
        \draw[dashed] (3, 0.5) -- (3, 2.5);
        \draw[dashed] (4, 0.5) -- (4, 1.5);
        \draw[dashed] (5, 0.5) -- (5, 2.5);
        \draw[dashed] (6, 0.5) -- (6, 1.5);
        
        \node[below] at (0,2) {\textbf{0}};
        \node[below] at (2,2) {\textbf{1}};
        \node[below] at (4,2) {\textbf{0}};
        \node[below] at (6,2) {\textbf{1}};
        
        \node[below] at (1,3) {\textbf{0}};
        \node[below] at (3,3) {\textbf{1}};
        \node[below] at (5,3) {\textbf{0}};
        \node[below] at (7,2.5) {$t = 0$};

        \node[below] at (0,5) {\textbf{1}};
        \node[below] at (2,5) {\textbf{0}};
        \node[below] at (4,5) {\textbf{1}};
        \node[below] at (6,5) {\textbf{0}};
        
        \node[below] at (1,4) {\textbf{0}};
        \node[below] at (3,4) {\textbf{1}};
        \node[below] at (5,4) {\textbf{0}};
        \node[below] at (7,5) {$t = 1$};
        
        \end{tikzpicture}
    \end{subfigure}
    \begin{subfigure}{0.475\textwidth}
        \centering
        \begin{tikzpicture}[font=\sffamily\small, scale=0.75, transform shape]]
        \draw[->, thick] (1.0,1.0) -- (8,1.0) node[below] {$\mathbb{Z}$};
        \draw[->, thick] (1.0,1.0) -- (1.0,5.5) node[left] {$t$};

        \draw[thick, pattern=horizontal lines light gray] (1,1.0) -- (1, 3.5) -- (1.5, 3) -- (2, 3.5) -- (2.5, 4) -- (3, 3.5) -- (4, 4.5) -- (4.5, 5) -- (5, 4.5) -- (6, 3.5) -- (6.5, 4) -- (7, 3.5) --(7.5, 3) -- (7.5, 1.0) -- cycle;
        
        \draw[dashed] (1, 1.0) -- (1, 3.5);
        \draw[dashed] (1.5, 1.0) -- (1.5, 3);
        \draw[dashed] (2, 1.0) -- (2, 3.5);
        \draw[dashed] (2.5, 1.0) -- (2.5, 4);
        \draw[dashed] (3, 1.0) -- (3, 3.5);
        \draw[dashed] (3.5, 1.0) -- (3.5, 4);
        \draw[dashed] (4, 1.0) -- (4, 4.5);
        \draw[dashed] (4.5, 1.0) -- (4.5, 5);
        \draw[dashed] (5, 1.0) -- (5, 4.5);
        \draw[dashed] (5.5, 1.0) -- (5.5, 4);
        \draw[dashed] (6, 1.0) -- (6, 3.5);
        \draw[dashed] (6.5, 1.0) -- (6.5, 4);
        \draw[dashed] (7, 1.0) -- (7, 3.5);
        
        \draw[dotted, semithick] (1, 3.5) -- (1.5, 4)--(2, 3.5);
        \draw[dotted, semithick] (2.5, 4) -- (3, 4.5)--(3.5, 4);
        \draw[dotted, semithick] (5.5, 4) -- (6, 4.5)--(6.5, 4);
        
        \draw[thick, >=stealth, ->] (1.5, 3) -- (1.5, 4);
        \draw[thick, >=stealth, ->] (3, 3.5) -- (3, 4.5);
        \draw[thick, >=stealth, ->] (6, 3.5) -- (6, 4.5);
        
        \end{tikzpicture}
    \end{subfigure}
    \caption{Representation of a flip automata network for \rle{212} as a mountain-valley landscape (left) and evolution of a mountain-valley landscape carrying information about the frequency of updates (right).}
    \label{fig:wave}
    \end{figure}

\end{example}

\begin{restatable}[Duality]{lemma}{causality}\label{lem:causality}
    Every ACA $\mathbb{A} = (S, N, f, \mathbb{Z}^d)$ with von Neumann neighborhood where no two neighboring cells can be simultaneously active has a ``dual'' FAN $\mathcal{A} = (T=\Z^d, S, \{f'\})$ such that every configuration of $\mathbb{A}$ can be projected to a valid configuration of $\mathcal{A}$, and for any update schedule $\zeta$ their evolutions commute with the projection.
\end{restatable}

The proof for this result is similar to \Cref{ex:mountain-valley} and can be found in \appendixref{app:ommited-proofs}. As a direct consequence we obtain that all such ACAs have invariant histories due to \cref{prop:flip-invariant} (can be alternatively shown using commutativity). \cref{lem:causality} effectively shows the duality between a large subclass of ACAs and FANs, which will prove useful in understanding their computational mechanics.

\begin{openquestion}
    Can duality between FANs and ACAs be extended to all commutative ACAs (perhaps with modifications to \cref{def:flip-automata-network})?
\end{openquestion}

\section{Simulating Synchronous Networks}\label{sec:simulation}

We start by discussing the notion of simulation between cellular automata and present its extension to the asynchronous scenario. For brevity, we refer to the simulated (synchronous) automaton as the ``guest'' CA, and the simulating (a)synchronous automaton as the ``host''.

\begin{definition}[Direct Simulation]
    A CA $B = (S^B, N^B, f^B, \mathbb{Z}^d)$ with global transition function $G^B$ directly simulates a CA $A = (S^A, N^A, f^A, \mathbb{Z}^d)$ with global transition function $G^A$ of the same dimension according to a mapping $\psi\colon S^A \longrightarrow \mathcal{P}(S^B) \setminus \{ \varnothing \}$ if for all states $a, b \in S^A$ we have $\psi(a) \cap \psi(b) = \varnothing$ and for any configuration $c \in {(S^A)}^{\Z^d}$ we have 
    \begin{equation*}
        \{ G^B(c') \ | \ c' \in \psi(c) \} \subseteq \psi(G^A(c)) \quad \text{where} \quad \psi(c) \coloneqq i \mapsto \psi(c(i))
    \end{equation*}
    i.e. all encodings $c'$ evolve in $B$ as $c$ evolves in $A$. In this case we write $G^A \prec G^B$.
\end{definition}

The general definition of simulation extends the definition of direct simulation by allowing packing, cutting, and shifting, such that a direct simulation can occur.

\begin{definition}[Unpacking Map]\label{def:direct-simulation}
    Let $S$ be a state set and $m = (m_1, \ldots, m_d)$ be a tuple of positive integers, then the unpacking bijective map
    \begin{equation*}
        o_m\colon{\big(S^{\prod m_i}\big)}^{\Z^d} \longrightarrow S^{\Z^d}
    \end{equation*}
    is defined for every configuration $c \in {\big(S^{\prod m_i}\big)}^{\Z^d}$ such that for every position $i \in \Z^d$ and offset $r \in \prod_i \Z_{m_i}$ it holds that $o_m(c)(m_1i_1 + r_1, \ldots, m_d i_d + r_d) = c(i)(r)$. 
\end{definition}

In essence, the unpacking map $o_m$ transforms every cell $i$ with state from $S^{\prod m_i}$ into a block of cells of size $\prod m_i$ and bijectively maps the state information from $i$ to the block. 

\begin{definition}[Simulation \cite{Ollinger-chapter}]\label{def:simulation}
    A cellular automaton $B = (S^B, N^B, f^B, \mathbb{Z}^d)$ simulates a  cellular automaton $A = (S^A, N^A, f^A, \mathbb{Z}^d)$ if for some $m$ there exists an unpacking map $o_m$, a translation $\tau_v$ and a positive integer $n \in \N$ such that
    \begin{equation*}
        {(G^A)} \prec o_m^{-1} \circ {(G^B)}^n \circ o_m \circ \tau_v
    \end{equation*}
    In this case, we say that $B$ simulates $A$ in $n$-linear time (and in real time if $n=1$).
\end{definition}

For the asynchronous scenario multiple simulation definitions have been proposed, however, all of them have turned out to be either too restrictive or flawed \cite{nakamura, improved-soldiers}. Most notably, they did not support for ``transitivity'', i.e., if an ACA $\mathbb{C}$ simulates a CA $B$, and $B$ in turn simulates another CA $A$, then one would expect $\mathbb{C}$ to simulate $A$ (which was not fulfilled). To mitigate these issues, we propose a natural extension of \Cref{def:simulation} based on the notion of synchronous simulation and invariant histories, such that it is general enough to encompass most of the previous simulation constructions \cite{nakamura, gacs} (see \appendixref{app:examples}), while preserving the crucial ``transitivity'' property.

\begin{definition}[Invariant Simulation]\label{def:invariant}
     Let $\mathbb{B} = (S^B, N^B, f^B, \mathbb{Z}^d)$ be an ACA and $A = (S^A, N^A, f^A, \mathbb{Z}^d)$ be some synchronous CA of the same dimension, and let
     \begin{equation*}
         H^B\colon{(S^B)}^{\Z^d} \longrightarrow {(S^B)}^{\Z^d}
     \end{equation*}
     be a function that transforms the configuration by mapping every cell $i$ in state $h_t(i)$ to the next point in its invariant update history $h_{t+1}(i)$ under $\mathbb{B}$ ($H^B(c)$ is undefined in case $\mathbb{B}$ does not have invariant histories for some $c$). We say that $\mathbb{B}$ invariantly simulates $A$ for some $m$ there exists an unpacking map $o_m$, a translation $\tau_v$ and positive integers $k, l \in \N$ such that 
     \begin{equation*}
        {(G^A)}^{k} \prec o_m^{-1} \circ {(H^B)}^l \circ o_m \circ \tau_v
    \end{equation*}
    where $l/k$ is again the linear time parameter.
\end{definition}

\cref{def:invariant} requires configurations of $A$ to appear as complete, encoded rows in the invariant history $h$ of $\mathbb{B}$, allowing the observer to reconstruct the evolution of $A$ from the evolution of $\mathbb{B}$, even if the update schedule $\zeta$ is unknown. The ``transitivity'' is also fulfilled, since we can chain unpackings and translations between different types of simulation.

\subparagraph*{Simulation techniques} The classical approach to asynchronous simulation requires the ``host'' ACA to be of the same dimension $d$ and neighborhood $N$ as the ``guest'' CA--indeed, when synchronizing a distributed system, the general goal is to alter the latter as little as possible. The classical $3q^2$ real time construction by Nakamura \cite{nakamura, gacs} (see \appendixref{app:examples}) and $q^2 + 2q$ state $3$-linear time construction in \cite{improved-soldiers, worsch-intrinsically} achieve this by increasing the number of states in the ``host'' (given a $q$-state ``guest'' CA) and use one cell of the ``host'' to simulate one cell of the ``guest''. Furthermore, these approaches require the neighborhood $N$ to be symmetric (this constraint will be investigated more in \cref{sec:universality}). Our new construction--based on embeddings of flip automata networks into ACAs--allows for simulations with both linear state and time overheads, which is achieved using spatial grouping of cells, as introduced in \cref{def:invariant}. We also restrict the neighborhood $N$ to be either von Neumann or Moore, since these constitute the most crucial cases.

\begin{proposition}[Mountain-Valley Encoding]\label{prop:mountain-valley}
    For a synchronous automaton $A = (S, N, f, \Z^d)$ with von Neumann or Moore's neighborhood $N$ we can construct an asynchronous cellular automaton $\mathbb{B} = (S \times \{0, 1, 2, 3\}, N, f', \Z^d)$ which simulates $A$ in real time.
\end{proposition}

\begin{proof}
    The high-level idea behind the construction of $\mathbb{B}$ is that we first construct a flip automata network $\mathcal{A}$ which simulates $A$, and then create an asynchronous cellular automaton embedding $\mathcal{A}$. To keep the construction elegant, we present it for the $2$-dimensional von Neumann neighborhood $N = ((-1, 0), (1, 0), (0, 0), (0, 1), (0, -1))$ in detail, and give a sketch for Moore's neighborhood. The generalization to any dimension $d$ is straightforward.
    
    Consider a FAN $\mathcal{A} = (T = \Z^2, S, \{f_\textbf{E}, f_\textbf{O}\})$ with the same set of states as $A$, where we distinguish between two types of nodes \textbf{E} and \textbf{O}, and set $\oplus / \ominus$ to be addition and subtraction mod $|S|$. Then, for von Neumann neighborhood:
    \begin{itemize}
        \item \textbf{E} nodes are placed on coordinates $(2i, 2j) \in \Z^2$ and we set
        \begin{equation*}
            f_\textbf{E}(s_{(1, 0)}, s_{(0, 1)}, s, s_{(-1, 0)}, s_{(0, -1)}) = f(s_{(1, 0)} \ominus s, s_{(0, 1)} \ominus s, s, s_{(-1, 0)} \ominus s, s_{(0, -1)} \ominus s)
        \end{equation*}
        Intuitively, \textbf{E} nodes store the states of the cells in $A$ and simulate the local transition $f$
        \item \textbf{O} nodes are placed at all remaining coordinates and every $f_\textbf{O}$ computes the sum $\oplus$ of all of its \textbf{E} neighbors. On a high level, \textbf{O} nodes use addition mod $|S|$ to accumulate and pass on the information about (simulated) neighbors of \textbf{E} nodes
    \end{itemize}
    Observe the topology of $\mathcal{A}$ as presented in \Cref{fig:mountain-valley}--the direction indicators are only added between neighboring \textbf{E} and \textbf{O} nodes. For von Neumann neighborhood (left), the \textbf{O} nodes with odd coordinates will be isolated in $\mathcal{A}$ since they do not have any \textbf{E} neighbors. On the other hand, in the case of Moore's neighborhood (right), the topology of $\mathcal{A}$ becomes non-uniform with different \textbf{O} nodes having different neighborhoods. We initialize $\mathcal{A}$ with all direction indicators pointing from \textbf{E} nodes to \textbf{O} nodes, and on every local transition, we flip all edges. Furthermore, the initial configuration $c_0$ of $A$ is encoded on the \textbf{E} nodes, where $T(2i) = c_0(i)$, and \textbf{O} nodes are initialized arbitrarily. If we first update all \textbf{O} nodes, flip the arrows, and update all \textbf{E} nodes, the state of the network will be an encoding of $G^A(c_0)$.
    
    Finally, let us construct an asynchronous cellular automaton $\mathbb{B} = (S \times \{0, 1, 2, 3\}, N, f', \Z^2)$ corresponding to FAN $\mathcal{A}$. First, let us define the initial configuration $c'_0$ on $\mathbb{B}$. For each \textbf{E} node at $(i, j)$ we set $c'_0(2i, 2j) = c_0(i, j)2$ (tuple notation), and each \textbf{O} node is initialized with $s1$ where $s \in S$ arbitrary (see \Cref{fig:mountain-valley}). We encode direction indicators of $\mathcal{A}$ in $\mathbb{B}$ by defining a mapping $\psi$ that relies on neighboring timers $t \in \{0, 1, 2, 3\}$ for each edge:
    \begin{align*}
        &\psi(t, t+1) = t \longleftarrow t+1 \pmod{4} &\psi(t+1, t) = t+1 \longrightarrow t \pmod{4}
    \end{align*}
    There are no arrows if $t$s differ by $0$ or by $2$, i.e., the node simply ``ignores'' such neighbors. We define the local rule function $f'$ to be active on cell $i$ if and only if all (implicit) arrows are pointing to $i$. More formally, if a cell in state $st$ has an even timer $t$ it mimics $f_\textbf{E}$ and updates $t$ by $2$, and if it has an odd timer it mimics $f_\textbf{O}$ and again, updates $t$ by $2$ (the extension to Moore's neighborhood is analogous):
    \begin{align*}
        f'(s_{(1, 0)}(t+1), s_{(0, 1)}(t+1), st, &s_{(-1, 0)}(t+1), s_{(0, -1)}(t+1)) = \\
        f_\textbf{E}(s_{(1, 0)}, s_{(0, 1)}, &s, s_{(-1, 0)}, s_{(0, -1)}) (t+2) \quad \text{if} \quad t \ \text{even}\\
        f'(s_{(1, 0)}t_0, s_{(0, 1)}t_1, st, s_{(-1, 0)}t_2, &s_{(0, -1)}t_3) = \\
        f_\textbf{O}(s_{(1, 0)}, s_{(0, 1)}, &s, s_{(-1, 0)}, s_{(0, -1)}) (t+2) \quad \text{if} \quad t \ \text{odd } \wedge t_i + 1 \neq t \  \forall i 
    \end{align*}
    \begin{figure}[tb]
    \centering
    \begin{tikzpicture}[font=\sffamily\small, scale=0.7, transform shape]
    \draw[->, thick] (0.0,0.0) -- (4.0,0.0) node[below] {$\Z$};
    \draw[->, thick] (0.0,0.0) -- (0.0,4.0) node[left] {$\Z$};

    \foreach \i in {0,...,5}{
        \draw[dashed] (0.75*\i, 0) -- (0.75*\i, 3.75);
        \draw[dashed] (0, 0.75*\i) -- (3.75, 0.75*\i);
    }

    \foreach \i in {0,...,2}{
        \draw[<-, color=darkred, thick, >=stealth] (0.375 + 1.5*\i, 0.6) -- (0.375 + 1.5*\i, 0.9);
        \draw[<-, color=darkred, thick, >=stealth] (0.375 + 1.5*\i, 2.1) -- (0.375 + 1.5*\i, 2.4);
        \draw[->, color=darkred, thick, >=stealth] (0.375 + 1.5*\i, 1.35) -- (0.375 + 1.5*\i, 1.65);
        \draw[->, color=darkred, thick, >=stealth] (0.375 + 1.5*\i, 2.85) -- (0.375 + 1.5*\i, 3.15);
    }  

    \foreach \i in {0,...,2}{
        \draw[<-, color=darkred, thick, >=stealth] (0.6, 0.375 + 1.5*\i) -- (0.9, 0.375 + 1.5*\i);
        \draw[<-, color=darkred, thick, >=stealth] (2.1, 0.375 + 1.5*\i) -- (2.4, 0.375 + 1.5*\i);
        \draw[->, color=darkred, thick, >=stealth] (1.35, 0.375 + 1.5*\i) -- (1.65, 0.375 + 1.5*\i);
        \draw[->, color=darkred, thick, >=stealth] (2.85, 0.375 + 1.5*\i) -- (3.15, 0.375 + 1.5*\i);
    }

    \node[] at (0.375,0.375) {$\textbf{E}$};
    \node[] at (1.125,0.375) {$\textbf{O}$};
    \node[] at (1.875,0.375) {$\textbf{E}$};
    \node[] at (2.625,0.375) {$\textbf{O}$};
    \node[] at (3.375,0.375) {$\textbf{E}$};

    \node[] at (0.375,1.125) {$\textbf{O}$};
    \node[] at (1.125,1.125) {$\textbf{O}$};
    \node[] at (1.875,1.125) {$\textbf{O}$};
    \node[] at (2.625,1.125) {$\textbf{O}$};
    \node[] at (3.375,1.125) {$\textbf{O}$};

    \node[] at (0.375,1.875) {$\textbf{E}$};
    \node[] at (1.125,1.875) {$\textbf{O}$};
    \node[] at (1.875,1.875) {$\textbf{E}$};
    \node[] at (2.625,1.875) {$\textbf{O}$};
    \node[] at (3.375,1.875) {$\textbf{E}$};

    \node[] at (0.375,2.625) {$\textbf{O}$};
    \node[] at (1.125,2.625) {$\textbf{O}$};
    \node[] at (1.875,2.625) {$\textbf{O}$};
    \node[] at (2.625,2.625) {$\textbf{O}$};
    \node[] at (3.375,2.625) {$\textbf{O}$};

    \node[] at (0.375,3.375) {$\textbf{E}$};
    \node[] at (1.125,3.375) {$\textbf{O}$};
    \node[] at (1.875,3.375) {$\textbf{E}$};
    \node[] at (2.625,3.375) {$\textbf{O}$};
    \node[] at (3.375,3.375) {$\textbf{E}$};

    \end{tikzpicture}
    \begin{tikzpicture}[font=\sffamily\small, scale=0.7, transform shape]
    
    \draw[->, thick] (0.0,0.0) -- (4.0,0.0) node[below] {$\Z$};
    \draw[->, thick] (0.0,0.0) -- (0.0,4.0) node[left] {$\Z$};

    \foreach \i in {0,...,5}{
        \draw[dashed] (0.75*\i, 0) -- (0.75*\i, 3.75);
        \draw[dashed] (0, 0.75*\i) -- (3.75, 0.75*\i);

    \foreach \i in {0,...,2}{
        \draw[<-, color=darkred, thick, >=stealth] (0.375 + 1.5*\i, 0.6) -- (0.375 + 1.5*\i, 0.9);
        \draw[<-, color=darkred, thick, >=stealth] (0.375 + 1.5*\i, 2.1) -- (0.375 + 1.5*\i, 2.4);
        \draw[->, color=darkred, thick, >=stealth] (0.375 + 1.5*\i, 1.35) -- (0.375 + 1.5*\i, 1.65);
        \draw[->, color=darkred, thick, >=stealth] (0.375 + 1.5*\i, 2.85) -- (0.375 + 1.5*\i, 3.15);
    }  

    \foreach \i in {0,...,2}{
        \draw[<-, color=darkred, thick, >=stealth] (0.6, 0.375 + 1.5*\i) -- (0.9, 0.375 + 1.5*\i);
        \draw[<-, color=darkred, thick, >=stealth] (2.1, 0.375 + 1.5*\i) -- (2.4, 0.375 + 1.5*\i);
        \draw[->, color=darkred, thick, >=stealth] (1.35, 0.375 + 1.5*\i) -- (1.65, 0.375 + 1.5*\i);
        \draw[->, color=darkred, thick, >=stealth] (2.85, 0.375 + 1.5*\i) -- (3.15, 0.375 + 1.5*\i);
    }  

    \node[] at (0.375,0.375) {$\text{0}$};
    \node[] at (1.125,0.375) {$\text{1}$};
    \node[] at (1.875,0.375) {$\text{0}$};
    \node[] at (2.625,0.375) {$\text{1}$};
    \node[] at (3.375,0.375) {$\text{0}$};

    \node[] at (0.375,1.125) {$\text{1}$};
    \node[] at (1.125,1.125) {$\text{1}$};
    \node[] at (1.875,1.125) {$\text{1}$};
    \node[] at (2.625,1.125) {$\text{1}$};
    \node[] at (3.375,1.125) {$\text{1}$};

    \node[] at (0.375,1.875) {$\text{0}$};
    \node[] at (1.125,1.875) {$\text{1}$};
    \node[] at (1.875,1.875) {$\text{0}$};
    \node[] at (2.625,1.875) {$\text{1}$};
    \node[] at (3.375,1.875) {$\text{0}$};

    \node[] at (0.375,2.625) {$\text{1}$};
    \node[] at (1.125,2.625) {$\text{1}$};
    \node[] at (1.875,2.625) {$\text{1}$};
    \node[] at (2.625,2.625) {$\text{1}$};
    \node[] at (3.375,2.625) {$\text{1}$};

    \node[] at (0.375,3.375) {$\text{0}$};
    \node[] at (1.125,3.375) {$\text{1}$};
    \node[] at (1.875,3.375) {$\text{0}$};
    \node[] at (2.625,3.375) {$\text{1}$};
    \node[] at (3.375,3.375) {$\text{0}$};
    }
    \end{tikzpicture}
    \begin{tikzpicture}[font=\sffamily\small, scale=0.7, transform shape]
    \draw[->, thick] (0.0,0.0) -- (4.0,0.0) node[below] {$\Z$};
    \draw[->, thick] (0.0,0.0) -- (0.0,4.0) node[left] {$\Z$};

    \foreach \i in {0,...,5}{
        \draw[dashed] (0.75*\i, 0) -- (0.75*\i, 3.75);
        \draw[dashed] (0, 0.75*\i) -- (3.75, 0.75*\i);
    }

    \foreach \i in {0,...,2}{
        \draw[<-, color=darkred, thick, >=stealth] (0.375 + 1.5*\i, 0.6) -- (0.375 + 1.5*\i, 0.9);
        \draw[<-, color=darkred, thick, >=stealth] (0.375 + 1.5*\i, 2.1) -- (0.375 + 1.5*\i, 2.4);
        \draw[->, color=darkred, thick, >=stealth] (0.375 + 1.5*\i, 1.35) -- (0.375 + 1.5*\i, 1.65);
        \draw[->, color=darkred, thick, >=stealth] (0.375 + 1.5*\i, 2.85) -- (0.375 + 1.5*\i, 3.15);
    }  

    \foreach \i in {0,...,2}{
        \draw[<-, color=darkred, thick, >=stealth] (0.6, 0.375 + 1.5*\i) -- (0.9, 0.375 + 1.5*\i);
        \draw[<-, color=darkred, thick, >=stealth] (2.1, 0.375 + 1.5*\i) -- (2.4, 0.375 + 1.5*\i);
        \draw[->, color=darkred, thick, >=stealth] (1.35, 0.375 + 1.5*\i) -- (1.65, 0.375 + 1.5*\i);
        \draw[->, color=darkred, thick, >=stealth] (2.85, 0.375 + 1.5*\i) -- (3.15, 0.375 + 1.5*\i);
    }

    \foreach \i in {0,...,1}{
        \foreach \j in {0,...,1}{
            \draw[<-, color=darkred, thick, >=stealth] (0.6 + 1.5*\i, 0.6 + 1.5*\j) -- (0.9 + 1.5*\i, 0.9+1.5*\j);
            \draw[->, color=darkred, thick, >=stealth] (1.35 + 1.5*\i, 1.35+1.5*\j) -- (1.65 + 1.5*\i, 1.65+1.5*\j);
            \draw[->, color=darkred, thick, >=stealth] (0.9 + 1.5*\i, 1.35 + 1.5*\j) -- (0.6 + 1.5*\i, 1.65+1.5*\j);
            \draw[<-, color=darkred, thick, >=stealth] (1.65 + 1.5*\i, 0.6+1.5*\j) -- (1.35 + 1.5*\i, 0.9+1.5*\j);
        }
    }

    \node[] at (0.375,0.375) {$\textbf{E}$};
    \node[] at (1.125,0.375) {$\textbf{O}$};
    \node[] at (1.875,0.375) {$\textbf{E}$};
    \node[] at (2.625,0.375) {$\textbf{O}$};
    \node[] at (3.375,0.375) {$\textbf{E}$};

    \node[] at (0.375,1.125) {$\textbf{O}$};
    \node[] at (1.125,1.125) {$\textbf{O}$};
    \node[] at (1.875,1.125) {$\textbf{O}$};
    \node[] at (2.625,1.125) {$\textbf{O}$};
    \node[] at (3.375,1.125) {$\textbf{O}$};

    \node[] at (0.375,1.875) {$\textbf{E}$};
    \node[] at (1.125,1.875) {$\textbf{O}$};
    \node[] at (1.875,1.875) {$\textbf{E}$};
    \node[] at (2.625,1.875) {$\textbf{O}$};
    \node[] at (3.375,1.875) {$\textbf{E}$};

    \node[] at (0.375,2.625) {$\textbf{O}$};
    \node[] at (1.125,2.625) {$\textbf{O}$};
    \node[] at (1.875,2.625) {$\textbf{O}$};
    \node[] at (2.625,2.625) {$\textbf{O}$};
    \node[] at (3.375,2.625) {$\textbf{O}$};

    \node[] at (0.375,3.375) {$\textbf{E}$};
    \node[] at (1.125,3.375) {$\textbf{O}$};
    \node[] at (1.875,3.375) {$\textbf{E}$};
    \node[] at (2.625,3.375) {$\textbf{O}$};
    \node[] at (3.375,3.375) {$\textbf{E}$};

    \end{tikzpicture}
    \begin{tikzpicture}[font=\sffamily\small, scale=0.7, transform shape]
    
    \draw[->, thick] (0.0,0.0) -- (4.0,0.0) node[below] {$\Z$};
    \draw[->, thick] (0.0,0.0) -- (0.0,4.0) node[left] {$\Z$};

    \foreach \i in {0,...,5}{
        \draw[dashed] (0.75*\i, 0) -- (0.75*\i, 3.75);
        \draw[dashed] (0, 0.75*\i) -- (3.75, 0.75*\i);

    \foreach \i in {0,...,2}{
        \draw[<-, color=darkred, thick, >=stealth] (0.375 + 1.5*\i, 0.6) -- (0.375 + 1.5*\i, 0.9);
        \draw[<-, color=darkred, thick, >=stealth] (0.375 + 1.5*\i, 2.1) -- (0.375 + 1.5*\i, 2.4);
        \draw[->, color=darkred, thick, >=stealth] (0.375 + 1.5*\i, 1.35) -- (0.375 + 1.5*\i, 1.65);
        \draw[->, color=darkred, thick, >=stealth] (0.375 + 1.5*\i, 2.85) -- (0.375 + 1.5*\i, 3.15);
    }  

    \foreach \i in {0,...,2}{
        \draw[<-, color=darkred, thick, >=stealth] (0.6, 0.375 + 1.5*\i) -- (0.9, 0.375 + 1.5*\i);
        \draw[<-, color=darkred, thick, >=stealth] (2.1, 0.375 + 1.5*\i) -- (2.4, 0.375 + 1.5*\i);
        \draw[->, color=darkred, thick, >=stealth] (1.35, 0.375 + 1.5*\i) -- (1.65, 0.375 + 1.5*\i);
        \draw[->, color=darkred, thick, >=stealth] (2.85, 0.375 + 1.5*\i) -- (3.15, 0.375 + 1.5*\i);
    }  

    \foreach \i in {0,...,1}{
        \foreach \j in {0,...,1}{
            \draw[<-, color=darkred, thick, >=stealth] (0.6 + 1.5*\i, 0.6 + 1.5*\j) -- (0.9 + 1.5*\i, 0.9+1.5*\j);
            \draw[->, color=darkred, thick, >=stealth] (1.35 + 1.5*\i, 1.35+1.5*\j) -- (1.65 + 1.5*\i, 1.65+1.5*\j);
            \draw[->, color=darkred, thick, >=stealth] (0.9 + 1.5*\i, 1.35 + 1.5*\j) -- (0.6 + 1.5*\i, 1.65+1.5*\j);
            \draw[<-, color=darkred, thick, >=stealth] (1.65 + 1.5*\i, 0.6+1.5*\j) -- (1.35 + 1.5*\i, 0.9+1.5*\j);
        }
    }  

    \node[] at (0.375,0.375) {$\text{0}$};
    \node[] at (1.125,0.375) {$\text{1}$};
    \node[] at (1.875,0.375) {$\text{0}$};
    \node[] at (2.625,0.375) {$\text{1}$};
    \node[] at (3.375,0.375) {$\text{0}$};

    \node[] at (0.375,1.125) {$\text{1}$};
    \node[] at (1.125,1.125) {$\text{1}$};
    \node[] at (1.875,1.125) {$\text{1}$};
    \node[] at (2.625,1.125) {$\text{1}$};
    \node[] at (3.375,1.125) {$\text{1}$};

    \node[] at (0.375,1.875) {$\text{0}$};
    \node[] at (1.125,1.875) {$\text{1}$};
    \node[] at (1.875,1.875) {$\text{0}$};
    \node[] at (2.625,1.875) {$\text{1}$};
    \node[] at (3.375,1.875) {$\text{0}$};

    \node[] at (0.375,2.625) {$\text{1}$};
    \node[] at (1.125,2.625) {$\text{1}$};
    \node[] at (1.875,2.625) {$\text{1}$};
    \node[] at (2.625,2.625) {$\text{1}$};
    \node[] at (3.375,2.625) {$\text{1}$};

    \node[] at (0.375,3.375) {$\text{0}$};
    \node[] at (1.125,3.375) {$\text{1}$};
    \node[] at (1.875,3.375) {$\text{0}$};
    \node[] at (2.625,3.375) {$\text{1}$};
    \node[] at (3.375,3.375) {$\text{0}$};
    }
    \end{tikzpicture}
    \caption{A FAN $\mathcal{A}$ with two types of nodes \textbf{E}, \textbf{O} and the encoding of $\mathcal{A}$ inside an ACA $\mathbb{B}$ using an even/odd timer $t \in \{0, 1, 2, 3\}$ with von Neumann neighborhood (left) and a similar FAN with identical encoding using Moore's neighborhood (right)}
    \label{fig:mountain-valley}
    \end{figure}
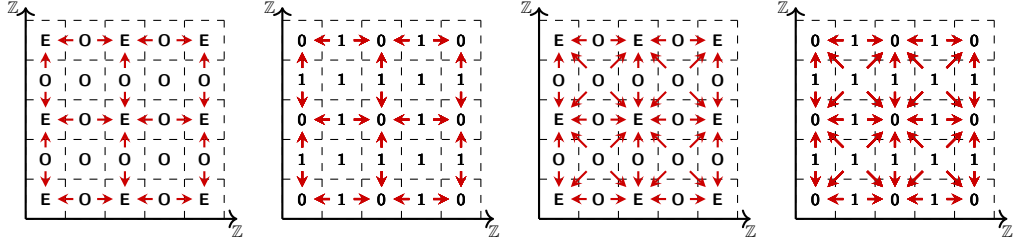
    Observe that simulated \textbf{E} cells will always have an even timer and \textbf{O} cells will always have an odd timer. Furthermore, every transition flips all local arrows by increasing the timer, thus ensuring that $\mathcal{A}$ is simulated correctly. Notice that all configurations in $\mathbb{B}$ which correspond to valid encodings of configurations of $A$ embed a FAN, hence will trivially have invariant histories. Let us confirm that our simulation satisfies \cref{def:invariant}. We can define a bijective unpacking map $o$ with
    \begin{equation*}
        o(c_0)(i, j) = (c_0'(2i, 2j), c_0'(2i+1, 2j), c_0'(2i, 2j+1), c_0'(2i+1, 2j+1)) 
    \end{equation*}
    which splits a cell of $c_0$ into $4$ different cells in $c'_0$ (one \textbf{E} cell and three \textbf{O} cells). Since $\mathcal{A}$ simulates $A$ using alternations between \textbf{E} and \textbf{O} we get $G^A \prec o^{-1} \circ {H^B} \circ o$. We conclude that $\mathbb{B}$ (asynchronously) simulates $A$ in real time.
\end{proof}

We can interpret the mountain-valley encoding as a $d$-dimensional mountain-valley landscape from \Cref{ex:mountain-valley}, which slowly travels through space-time, but now, each ``mountain'' and each ``valley'' stores the state information too. The results from \cref{prop:mountain-valley} can be improved from $4q$ to $3q$ for $1$-dimensional automata by introducing a translation inside the simulation. Notice that Moore and von Neumann neighborhood coincide for $1$-dimensional automata and are both simply first neighbors $N = (-1, 0, 1)$.

\begin{proposition}[Shifting Mountain-Valleys]\label{prop:improvement}
    For a synchronous automaton $A = (S, N, f, \Z)$ with a first neighbors neighborhood we can construct an asynchronous cellular automaton $\mathbb{B} = (S \times \{1, 2, 3\}, N, f', \Z)$ which simulates $A$ in $1.5$-linear time.
\end{proposition}

\begin{proof}
    The high-level idea is to improve the construction from \cref{prop:mountain-valley} by allowing \textbf{E} and \textbf{O} nodes to shift in space, which allows us to eliminate one step of our timer (at the cost of slowing down the simulation). Let us consider the mountain-valley landscape as in \Cref{ex:mountain-valley}, which we visualized as a zigzag line in space-time. Instead of using $2$ states, we will now use a tertiary timer to encode ``mountains'' and ``valleys'', and disallow nodes with equal values to be neighbors in our initial configuration:
    \begin{center}
       $(2, 1, 2) \mapsto 3 \quad\quad (3, 2, 3) \mapsto 1 \quad\quad (1, 3, 1) \mapsto 2$
    \end{center}
    The resulting cellular automaton will still have a mountain-valley structure if launched on periodic configuration $212121212$  (see \Cref{fig:tertiary-wave}). Now let $\mathbb{B} = (S \times \{1, 2, 3\}, N, f', \Z)$ be an asynchronous cellular automaton with a timer as described above, and where each state additionally carries simulated state $s \in S$. We encode the transitions of $A$ into $\mathbb{B}$ using the following update function $f'$ ($\oplus / \ominus$ denotes addition and subtraction mod $|S|$):
    \begin{align*}
        f'(s_{-1} 2, s_{0} 1, s_{1} 2) & = (s_{-1} \oplus s_{1})3, \\
        f'(s_{-1} 3, s_{0} 2, s_{1} 3) & = f(s_{-1} \ominus s_{0}, s_{0}, s_{1}, \ominus s_{0})1, \\
        f'(s_{-1} 1, s_{0} 3, s_{1} 1) & = s_{1}2.
    \end{align*}
    This execution paradigm can be seen as a mountain-valley landscape where each node has additionally a stage indicator. The stages (coinciding with timer values!) operate as follows:
    \begin{description}
        \item[Stage 1.] Compute the $\oplus$ sum of both of your neighbors and move to Stage 3
        \item[Stage 2.] Compute transition $f$ based on accumulated information and move to Stage 1
        \item[Stage 3.] Copy state from your right neighbor and move to Stage 2
    \end{description}
    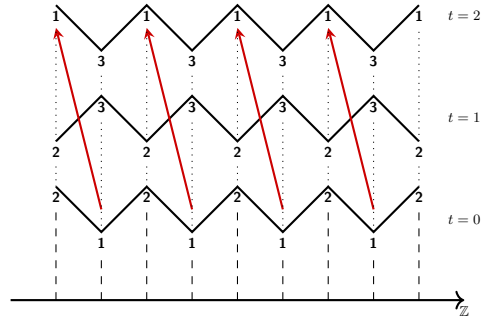
\begin{figure}[ht]
    \centering
    \begin{tikzpicture}[font=\sffamily\small, scale=0.6, transform shape]
    
    \draw[->, thick] (-2,0.5) -- (8,0.5) node[below] {$\mathbb{Z}$};

    \draw[black, thick] (-1,3) -- (0,2) -- (1, 3) -- (2, 2) -- (3, 3) -- (4, 2) -- (5, 3) -- (6, 2) -- (7, 3);
    \draw[black, thick] (-1, 4) -- (0,5) -- (1, 4) -- (2, 5) -- (3, 4) -- (4, 5) -- (5, 4) -- (6, 5) -- (7, 4);
    \draw[black, thick] (-1,7) -- (0,6) -- (1, 7) -- (2, 6) -- (3, 7) -- (4, 6) -- (5, 7) -- (6, 6) -- (7, 7);
    \draw[color=darkred, thick, ->, >=stealth,] (0,2.5)--(-1,6.5);
    \draw[color=darkred, thick, ->, >=stealth,] (2,2.5)--(1,6.5);
    \draw[color=darkred, thick, ->, >=stealth,] (4,2.5)--(3,6.5);
    \draw[color=darkred, thick, ->, >=stealth,] (6,2.5)--(5,6.5);
    
    \draw[dotted] (-1, 3) -- (-1, 3.5);
    \draw[dotted] (0, 2) -- (0, 4.5);
    \draw[dotted] (1, 3) -- (1, 3.5);
    \draw[dotted] (2, 2) -- (2, 4.5);
    \draw[dotted] (3, 3) -- (3, 3.5);
    \draw[dotted] (4, 2) -- (4, 4.5);
    \draw[dotted] (5, 3) -- (5, 3.5);
    \draw[dotted] (6, 2) -- (6, 4.5);
    \draw[dotted] (7, 3) -- (7, 3.5);

    \draw[dotted] (-1, 4) -- (-1, 6.5);
    \draw[dotted] (0, 5) -- (0, 5.5);
    \draw[dotted] (1, 4) -- (1, 6.5);
    \draw[dotted] (2, 5) -- (2, 5.5);
    \draw[dotted] (3, 4) -- (3, 6.5);
    \draw[dotted] (4, 5) -- (4, 5.5);
    \draw[dotted] (5, 4) -- (5, 6.5);
    \draw[dotted] (6, 5) -- (6, 5.5);
    \draw[dotted] (7, 4) -- (7, 6.5);

    \draw[dashed] (-1, 0.5) -- (-1, 2.5);
    \draw[dashed] (0, 0.5) -- (0, 1.5);
    \draw[dashed] (1, 0.5) -- (1, 2.5);
    \draw[dashed] (2, 0.5) -- (2, 1.5);
    \draw[dashed] (3, 0.5) -- (3, 2.5);
    \draw[dashed] (4, 0.5) -- (4, 1.5);
    \draw[dashed] (5, 0.5) -- (5, 2.5);
    \draw[dashed] (6, 0.5) -- (6, 1.5);
    \draw[dashed] (7, 0.5) -- (7, 2.5);

    \node[below] at (0,2) {\textbf{1}};
    \node[below] at (2,2) {\textbf{1}};
    \node[below] at (4,2) {\textbf{1}};
    \node[below] at (6,2) {\textbf{1}};

    \node[below] at (-1,3) {\textbf{2}};
    \node[below] at (1,3) {\textbf{2}};
    \node[below] at (3,3) {\textbf{2}};
    \node[below] at (5,3) {\textbf{2}};
    \node[below] at (7,3) {\textbf{2}};
    \node[below] at (8,2.5) {$t = 0$};

    \node[below] at (0,5) {\textbf{3}};
    \node[below] at (2,5) {\textbf{3}};
    \node[below] at (4,5) {\textbf{3}};
    \node[below] at (6,5) {\textbf{3}};

    \node[below] at (-1,4) {\textbf{2}};
    \node[below] at (1,4) {\textbf{2}};
    \node[below] at (3,4) {\textbf{2}};
    \node[below] at (5,4) {\textbf{2}};
    \node[below] at (7,4) {\textbf{2}};
    \node[below] at (8,4.75) {$t = 1$};

    \node[below] at (0,6) {\textbf{3}};
    \node[below] at (2,6) {\textbf{3}};
    \node[below] at (4,6) {\textbf{3}};
    \node[below] at (6,6) {\textbf{3}};

    \node[below] at (-1,7) {\textbf{1}};
    \node[below] at (1,7) {\textbf{1}};
    \node[below] at (3,7) {\textbf{1}};
    \node[below] at (5,7) {\textbf{1}};
    \node[below] at (7,7) {\textbf{1}};
    \node[below] at (8,7) {$t = 2$};
    
    \end{tikzpicture}
    \caption{Two synchronous updates of mountain valley landscape represented using a tertiary alternating timer $t \in \{1, 2, 3\}$. Notice that after $3$ time steps, the landscape shifts one cell to the left.}
    \label{fig:tertiary-wave}
    \end{figure}
    Now consider some initial configuration $c_0 = \ldots0100110\ldots$ of $A$ which we encode in $\mathbb{B}$ as $c'_0$ where $c_0'(2i) = c_0(i)2$ and $c_0'(2i+1) = c_0(i)1$. Observe that after two updates of all cells with odd coordinates and a single update of all cells with even coordinates, we obtain the encoded configuration $G^A(c_0)$ where the updated states are now placed at positions $2i$ and $2i-1$. Notice, $\mathbb{B}$ cannot have two neighboring transitions active at the same time, hence by \Cref{lem:causality}, it has invariant histories. By taking the unpacking map $o$ with $o(c_0)(i) = (c_0'(2i), c_0'(2i+1))$ and translation $\tau$ with $\tau(c)(i) = c(i-1)$ we observe:
    \begin{equation*}
        {(G^A)}^2 \prec o^{-1} \circ {(H^B)}^3  \circ o \circ \tau^2
    \end{equation*}
    which corresponds to $1.5$-linear time simulation.
\end{proof}

\section{Universality Frontier}\label{sec:universality}

Universality is one of the central concepts in computability theory, originally defined in the context of Turing machines \cite{turing}. For cellular automata there exist two distinct notions of universality: standard universality, i.e., the ability to ``simulate'' any Turing machine and \emph{intrinsic universality}, i.e., the ability to ``simulate'' any cellular automaton (a strictly stronger notion \cite{non-universal}). Since, contrary to Turing machines, cellular automata are both infinite and non-halting by design, a satisfactory definition for universality has always been challenging \cite{Ollinger-chapter}. We will use the widely accepted definition by Durand \cite{durand-gol}, which relies on \emph{ultimately periodic configurations} to perform computations (see \appendixref{sec:app-universality} for more details). The exact definition for ultimate periodicity varies across literature \cite{durand-gol, Ollinger-chapter, morita-book}, hence we present a version general enough to encompass most of them--in particular, the definition below allows to embed $d$-dimensional configurations into $(d+1)$-dimensional ones, leading to more generalizable universality constructions.

\begin{definition}[Ultimately Periodic Configuration]\label{def:ultimately-periodic}
    A configuration $c \in S^{\Z^d}$ is called ultimately periodic if there exists a $k \in \N$ and a set of scaled standard basis vectors $\vec{r}_1, \vec{r}_2, \ldots, \vec{r}_{d} \in \N^{d}$ such that for any $j$ and any cell $i = (i_1, \ldots, i_d) \in \Z^d$ with $i_j > k$ it holds $c(i + \vec{r_j}) = c(i)$.
\end{definition}

\subparagraph*{Synchronous frontier} Determining the ``computation boundary'' between universality and non-universality was originally considered by Shannon \cite{Shannon}, who initiated the search for the smallest universal Turing machine \cite{small-utms}. These questions also made their way to cellular automata, with the first $2$-dimensional universal cellular automaton being constructed by von Neumann \cite{vonneumann}. Banks \cite{Banks} has improved the number of states to just $2$ by showing that embedding infinite circuits inside a cellular space $\Z^2$ is sufficient for universality. Later, Minsky showed that Conway's \gol{} cellular automaton is universal \cite{gol}, and Cook \cite{rule-110} has constructed a universality proof for \rle{110}, confirming the existence of a $2$-state first neighbors automaton. On the other hand, Wolfram \cite{Wolfram2002} extensively studied elementary cellular automata and concluded that no one-way $2$-state universal automata exist (similarly, none exist with just one state or empty neighborhood). He also proposed a candidate $3$-state one-way cellular automaton; however, its universality is yet to be proven. The current ``universality frontier'' for synchronous CAs is given in \Cref{fig:asynchronous-frontier} (left).

Although an explicit definition of asynchronous universality is rarely presented in the literature, all important constructions of universal ACAs essentially assume the slight adjustment of the synchronous definition to the asynchronous scenario \cite{adachi, schneider, lee, lee-5}. This is consistent with our expectations--by using any simulation technique following \Cref{def:simulation} on any universal CA, we obtain a universal ACA. Below we present our formalization attempt but, like any other such attempt, it should be taken with a grain of salt. 

We call a configuration $c \in S^{\Z^d}$ \emph{computable} if the function $c \colon \Z^d \longrightarrow S$ is computable, and denote with $\lceil c \rceil$ the description of some Turing machine computing $c$ by returning an element of $S$ for every $d$-dimensional input vector. Similarly, for any description $r$ we denote with $\widetilde{r}$ the underlying configuration. We set $\bigl\lceil S^{\Z^d} \bigl\rceil$ as the space of all computable descriptions of $S^{\Z^d}$, and observe that it is closed under finite (a)synchronous evolution with any $G^\zeta$.

\begin{definition}[Asynchronous Universality]\label{def:asynchronous-universality}
    An ACA $\mathbb{A} = (S, N, f, \Z^d)$ is called universal if for some partial computable universal function $U \colon {\{0, 1\}}^\ast \longrightarrow {\{0, 1\}}^\ast \cup \{\bot\}$ we can find an ultimately periodic configuration $c^\infty$ (compiler) and total computable functions $e \colon {\{0, 1 \}}^\ast \longrightarrow \bigl\lceil S^{\Z^d} \bigl\rceil$ (encoder) and $d \colon \bigl\lceil S^{\Z^d} \bigl\rceil \longrightarrow {\{0, 1 \}}^\ast$ (decoder), a pattern $p \colon D \longrightarrow S$ with $D \subset \N \times \Z^{d}$ finite (halting condition) such that for every input $w \in {\{0, 1\}}^\ast$:
     \begin{enumerate}
         \item The initial configuration $c_0 = \widetilde{e(w)}$ is a finite perturbation of $c^\infty$, i.e. $\| c_0 - c^\infty \| < \infty$ where  $\| \cdot \|$ is the $l_0$ norm on  $S^{\Z^d}$.
         \item If $U(w)$ is defined, then for any update schedule $\zeta$ there exists the earliest (asynchronous) time stamp $t \in \N$ such that $p$ appeared in the observed slice of update history $h$ under $\zeta$, and the current configuration $c^\zeta_t$ satisfies $d(\lceil c^\zeta_t \rceil) = U(w)$.
        \item If $U(w) = \bot$ then $p$ will never appear in the update history $h$ under any choice of $\zeta$.
     \end{enumerate}
\end{definition}

\subparagraph*{Asynchronous frontier} Lee \cite{lee-5} kicked off the search for universal asynchronous automata by constructing a $5$-state universal von Neumann cellular automaton, which was later improved to $4$ states by Lee \cite{lee} and to $3$ states for Moore's neighborhood by Schneider \cite{schneider}. By allowing the neighborhood $N$ to be of radius $2$ (non-standard case), a $2$-state universal automaton was constructed by Adachi \cite{adachi}. On the other hand, by simulating \rle{110} using improved marching soldiers \cite{worsch-intrinsically, improved-soldiers}, an $8$-state first neighbors universal automaton can be inferred. We improve the frontier by constructing a $3$-state von Neumann and a $6$-state first neighbors universal automata (presented in \Cref{sec:small-aca}). The updated asynchronous universality frontier is given in \Cref{fig:asynchronous-frontier} (right).

\subparagraph*{One-way automata} Observe that all simulations mentioned in \Cref{sec:simulation} require the neighborhood $N$ to be symmetric, and in fact, no general simulation technique which is applicable to one-way neighborhoods is known. In \Cref{sec:non-universal} we give an explanation for this by proving that one-way ACAs are, in a sense ``equivalent'' to finite-state machines. This stands in sharp contrast to the synchronous case, where a simple procedure is known to embed any first neighbors automata into one-way automata (see \appendixref{app:examples}). Thus, we can make the following--somewhat unexpected--observation: the choice of asynchronous setting directly affects the computational capabilities of one-way ACAs. We summarize updated results on the computational power of ACAs under various asynchronous settings in \Cref{fig:one-way-universality}. 

\begin{figure}[thb]
    \centering
    \begin{tikzpicture}[scale=0.9, transform shape]
        \draw[->, thick] (0,0) -- (11,0) node[right, align=center, font=\scriptsize] {\textsf{computational}\\ \textsf{power}};
        \draw[fill=darkred] (10,0) circle[radius=1mm]; 
        \draw[] (10,0) node[above, align=center, font=\scriptsize] {``best-case'' \cite{aca-survey} \\  $\zeta$ is self-chosen \\}; 
        \draw[fill=darkred] (7.5,0) circle[radius=1mm]; 
        \draw[] (7.5,0) node[above, align=center, font=\scriptsize] {``classical'' \cite{ollinger-4-states} \\ $\zeta$ is synchronous \\}; 
        \draw[fill=darkred] (1,0) circle[radius=1mm]; 
        \draw[] (1,0) node[above, align=center, font=\scriptsize] {``worst-case'' $\bigstar$ \\ $\zeta$ is adversarial \\}; 
        \draw[] (4.25,0) node[above, align=center, font=\scriptsize] {``randomized'' \\ $\zeta$ is $\alpha$-asynchronous \\ \small \textbf{?}}; 
        \draw (1.25,-0.1)--(1.25,0.1);
        \draw (1.25, 0) node[below, align=center, font=\scriptsize] { \\ \textsf{ Finite-state} \\ \textsf{machines}};
        \draw (6,-0.1)--(6,0.1);
        \draw (6, 0) node[below, align=center, font=\scriptsize] { \\ \textsf{Turing} \\ \textsf{machines}};
    \end{tikzpicture}
    \caption{Computational power of one-way ACAs varies with respect to the choice for asynchronous setting. Notice that there are instances of intrinsically universal synchronous one-way CAs, hence they can be seen as more powerful than Turing machines.}
    \label{fig:one-way-universality}
\end{figure}
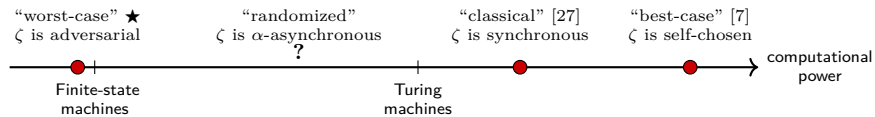

\begin{openquestion}
    What is the computational power of $\alpha$-asynchronous one-way ACAs?
\end{openquestion}

\section{Non-Universality of One-Way Automata}\label{sec:non-universal}

\begin{observation}\label{clm:non-universal}
    Asynchronous one-way cellular automata are computationally equivalent to finite-state machines, hence are not universal.
\end{observation}

Since there are many universal synchronous one-way automata (see \Cref{fig:asynchronous-frontier}), as an immediate consequence, we get the following result.

\begin{theorem}\label{thm:one-way}
    Asynchronous one-way cellular automata are not capable of simulating synchronous one-way cellular automata.
\end{theorem}

We prove \cref{thm:one-way,clm:non-universal} by constructing the required embeddings. \Cref{lem:non-universal} demonstrates that a finite-state machine can replace the evolution of a one-way asynchronous automaton, while \Cref{lem:semiautomaton} shows the converse by embedding an arbitrary finite-state transducer into an asynchronous one-way cellular automaton (see \appendixref{app:ommited-proofs}). 

\begin{lemma}\label{lem:non-universal}
    Given a one-way asynchronous cellular automaton $\mathbb{A} = (S, (-1, 0), f, \Z)$ we can find an update schedule $\zeta$ such that the computation of $\mathbb{A}$ on any input $w$ can be simulated by a semi-automaton $B = (S, S \cup \{ \# \}, \delta)$ with $S$ states and alphabet $S \cup \{\#\}$.
\end{lemma}

\begin{proof}
    We assume that we are given an ultimately periodic initial configuration $c_0$ with the input $w$ encoded into it. Since $\mathbb{A}$ is $1$-dimensional we have
    \begin{equation*}
        c_0 = \ldots p_Lp_Lp_L \hat{w} p_Rp_Rp_R\ldots
    \end{equation*}
    for some $p_L$ and $p_R$ independent of $w$, and $\hat{w}$, an encoding of $w$. Now we construct an update schedule $\zeta$ and a finite-state machine $B$ based on $p_L$ and $p_R$ and $\mathbb{A}$ such that $B$ simulates $\mathbb{A}$. Let $c$ be a configuration, $G$ (asynchronous) global transition function, $D \subseteq \Z$ be an update set, then we say that the state $s$ is a \emph{pseudo-fixed-point} of cell $i$ under $D$ if $s$ reappears infinitely often in the range of $G_D$, i.e., $|\{ k \in \N \ | \ G_D^k(c)(i) = s \}| = \infty$.

    We claim that for any cell $i$, configuration $c$ and update set $D$ there exists a pseudo-fixed-point $s$. Indeed, since $|S|$ is finite but the orbit of $i$ is infinite, at least one element $s \in S$ has to occur infinitely often. Now, consider the scenario where the update set $D = \{i\}$. We can determine the pseudo-fixed-point of $i$ on $c$ by considering the sequence
    \begin{equation*}
        c(i), \ f(c(i-1), c(i)), \ f(c(i-1), f(c(i-1), c(i))), \ \ldots
    \end{equation*}
    It is easy to see that if some state $s'$ appears twice in this sequence, it will appear infinitely often; hence it will be a pseudo-fixed-point. Since this simple pseudo-fixed-point only depends on the states $c(i-1)$ and $c(i)$ we will denote it by $\pf(c(i-1), c(i))$. Using this notion, we define the transition function $\delta$ for automaton $B$ and $q, s \in S$ as $\delta(q, s) = \pf(q, s)$.
    
    Let $l+1$ be the first cell that contains an element of $\hat{w}$ in $c_0$ and let $r-1$ be the last cell containing an element from $\hat{w}$ (see \Cref{fig:fixpoint}). Now set $x_l$ to be a pseudo-fixed-point of $l$ on configuration $c_0$ and update set $D = \{-\infty,\ldots,  l\}$ and for $j \in \{l+1, \ldots, r-1\}$ let $x_j = \pf(x_{j-1}, c_0(j))$. Observe that if we choose the initial state $x_l$ for $B$ and execute it on the first $j-l$ letters of $\hat{w}$, then $B$ will be in the state $x_j$! We construct the update schedule $\zeta$ by repeating steps 1--3 shown in \Cref{fig:fixpoint} indefinitely (every cell is updated infinitely often):
    \begin{enumerate}
        \item Repeatedly update the set $\{-\infty, \ldots, l\}$ until cell $l$ is in state $x_l$ 
        \item Repeatedly update $\{l+1\}$ until cell $l+1$ is in state $x_{l+1}$, then repeatedly update $\{l+2\}$ until cell $l+2$ is in state $x_{l+2}$, and so on until you reach $x_r$
        \item Update the set $\{r,\ldots, \infty \}$ once
    \end{enumerate}
    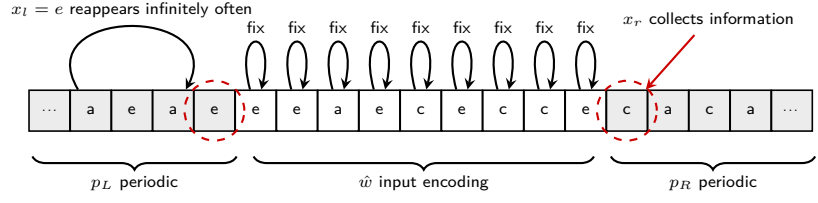
\begin{figure}[t]
    \centering
    \begin{tikzpicture}[font=\sffamily\scriptsize, scale=0.9, transform shape]
        \edef\sizetape{0.6cm}
        \tikzstyle{every path}=[semithick]
        \tikzstyle{tmtape}=[draw,minimum size=\sizetape]
        \tikzstyle{tmhead}=[arrow box,draw,minimum size=\sizetape,arrow box arrows={west:.3cm, east:0.3cm}]
        
        \begin{scope}[start chain=1 going right,node distance=-0.15mm]
            \node [on chain=1,tmtape, fill=lightgray!30] {$\scriptstyle \ldots$};
            \node [on chain=1,tmtape, fill=lightgray!30] (left1) {a};
            \node [on chain=1,tmtape, fill=lightgray!30] {e};
            \node [on chain=1,tmtape, fill=lightgray!30] (left2) {a};
            \node [on chain=1,tmtape, fill=lightgray!30] (begin){e};
            \node [on chain=1,tmtape] (1){e};
            \node [on chain=1,tmtape] (2){e};
            \node [on chain=1,tmtape] (3){a};
            \node [on chain=1,tmtape] (4){e};
            \node [on chain=1,tmtape] (5){c};
            \node [on chain=1,tmtape] (6){e};
            \node [on chain=1,tmtape] (7){c};
            \node [on chain=1,tmtape] (8){c};
            \node [on chain=1,tmtape] (9){e};
            \node [on chain=1,tmtape, fill=lightgray!30] (end){c};
            \node [on chain=1,tmtape, fill=lightgray!30] (right1){a};
            \node [on chain=1,tmtape, fill=lightgray!30] {c};
            \node [on chain=1,tmtape, fill=lightgray!30] (right2) {a};
            \node [on chain=1,tmtape, fill=lightgray!30] {$\scriptstyle \dots$};
        \end{scope}
        \draw [decorate,decoration={brace,amplitude=5pt,mirror,raise=4ex}, thick]
            (-0.25,0) -- (2.75,0) node[midway,yshift=-3em]{$p_L$ periodic};
        \draw [decorate,decoration={brace,amplitude=5pt,mirror,raise=4ex}, thick]
            (3,0) -- (8,0) node[midway,yshift=-3em]{$\hat{w}$ input encoding};
        \draw [decorate,decoration={brace,amplitude=5pt,mirror,raise=4ex}, thick]
            (8.25,0) -- (11.25,0) node[midway,yshift=-3em]{$p_R$ periodic};
            
        \node[above right = 0.8cm and -0.5cm of end] (info) {$x_r$ collects information};
        \node[circle, thick, color=darkred, draw, dashed, inner xsep=0mm, inner ysep = 0mm, fit=(begin)] {};
        \node[circle, thick, color=darkred, draw, dashed, inner xsep=0mm, inner ysep = 0mm, fit=(end)] {};

        \draw (left1) edge[bend left=4.25cm, looseness=2.35, above, thick, ->, >=stealth] node{$x_l = e$ reappears infinitely often} (left2);
        
        \draw (info) edge[color=darkred, thick, ->, >=stealth] node{} (end);
        \draw (1) edge[loop above, looseness=15, thick, ->, >=stealth] node{fix} (1);
        \draw (2) edge[loop above, looseness=15, thick, ->, >=stealth] node{fix} (2);
        \draw (3) edge[loop above, looseness=15, thick, ->, >=stealth] node{fix} (3);
        \draw (4) edge[loop above, looseness=15, thick, ->, >=stealth] node{fix} (4);
        \draw (5) edge[loop above, looseness=15, thick, ->, >=stealth] node{fix} (5);
        \draw (6) edge[loop above, looseness=15, thick, ->, >=stealth] node{fix} (6);
        \draw (7) edge[loop above, looseness=15, thick, ->, >=stealth] node{fix} (7);
        \draw (8) edge[loop above, looseness=15, thick, ->, >=stealth] node{fix} (8);
        \draw (9) edge[loop above, looseness=15, thick, ->, >=stealth] node{fix} (9);
    \end{tikzpicture}
    \caption{Evolution of $\mathbb{A}$ on $c_0$ under $\zeta$. Region $p_L$ cycles between the pseudo-fixed-point of $x_l$, and, so does every cell in $\hat{w}$, hence the cells in $p_R$ only receive constant information on $\hat{w}$.}
    \label{fig:fixpoint}
    \end{figure}
    We make the following observations about the space-time evolution of $c_0$ under $\zeta$: 
    \begin{itemize}
        \item After each iteration of steps 1--3 in $\zeta$ the states at positions $l, \ldots, r-1$ do not change and correspond to values $x_{l}, \ldots, x_{r-1}$
        \item Evolution of $p_L$ region is independent of $\hat{w}$ and $p_R$ region accesses $\hat{w}$ only through $x_{r-1}$
    \end{itemize}
    Since all states $x_{l+1}, \ldots, x_{r-1}$ are computed by $B$ on $\hat{w}$ and the update history $h$ of cells $\{r, \ldots, \infty \}$ only depends on $x_{r-1} \in S$ and $p_R$ (fixed), the right side of any configuration is only aware of bounded amount of information about $\hat{w}$. Thus, for any possible state of $x_{r-1}$, we can add a single transition $\delta(x_{r-1}, \#)$ to $B$, which represents the ``outcome'' of the joint evolution of $p_R$ and $x_{r-1}$. We conclude that the $\zeta$ evolution of $c_0$ based on the choice of $\hat{w}$ is at most as complex as the application of $B$ to $\hat{w}\#$, hence it cannot be universal.
\end{proof}

\begin{restatable}{lemma}{semiautomaton}\label{lem:semiautomaton}
    For any finite-state automaton $A =  (Q, \Sigma, \delta, q_0, F)$ there exists a one-way ACA $\mathbb{B} = (Q \times (\Sigma \cup \{ \varnothing \}) \cup \Sigma , (-1, 0), f, \mathbb{Z})$ simulating $A$. 
\end{restatable}

\section{Small Universal Asynchronous Automata}\label{sec:small-aca}

There have been two main approaches for designing small universal ACAs: simulating a synchronous CA or designing a universal system using inherently asynchronous primitives. 

\begin{theorem}[Six-State Universality]
    There exists a $6$-state first neighbors universal asynchronous cellular automaton.
\end{theorem}

\begin{proof}
    It is sufficient to apply our improved mountain-valley simulation from \Cref{prop:improvement} to \rle{110}. Since the corresponding invariant history will contain the synchronous space-time of \rle{110}, by extension, the resulting automaton is universal.
\end{proof}

\begin{theorem}[Three-State Universality]\label{thm:3-state}
    There exists a $3$-state von Neumann universal asynchronous cellular automaton.
\end{theorem}

To prove \Cref{thm:3-state} we use the second approach and design a $3$-state universal ACA using infinitely periodic delay-insensitive Boolean circuits. This kind of universality proof is known as \emph{circuit universality} and is commonly used to construct some of the smallest universal cellular automata. Due to the classification of Boolean functions by Post \cite{Post-boolean}, we know that any Boolean function $f$ can be implemented as a feed-forward circuit consisting solely of \textsf{NAND} gates (and the fan-out operation). This result can be further strengthened for \emph{planar} circuits (i.e., circuits without wire crossings).

\begin{proposition}[Planar Circuits \cite{goldschlager}]\label{thm:planar}
    The subset of planar feed-forward circuits implemented using only \textup{\textsf{NAND}} gates implements all Boolean functions.
\end{proposition}

It turns out that the ability to embed arbitrary circuits into $2$-dimensional cellular automata is already sufficient for universality and was used to prove universality of several CAs, most notably of \textsf{Life without Death} \cite{lifewithoutdeath}.

\begin{restatable}[Circuit Universality \cite{circuit-unviersality} \cite{lifewithoutdeath}]{proposition}{embedding}\label{prop:logical-universality}
    A $2$-dimensional CA which can implement any Boolean function on a finite subset of cells can embed the space-time of any first neighbors CA, where both use an ultimately periodic initial configuration, hence is universal.
\end{restatable}

The proof for \cref{prop:logical-universality} is given in \appendixref{app:ommited-proofs}, and it naturally generalizes to an asynchronous setting, as long as the ACA allows for reliable asynchronous implementation of Boolean functions. Universality proofs in \cite{lee-5, lee, adachi, schneider} accomplished this by constructing custom \emph{delay-insensitive} gates, i.e. gates that do not require simultaneity of inputs and operate under arbitrary delays \cite{di-gates}. For asynchronous inputs, a common tactic is to use \emph{dual-rail} encoding, i.e., encode each variable $\texttt{I}$ using two ``sub-wires'' $\texttt{I}$ and $\neg \texttt{I}$ (one for ``true'' signals, and the other for ``false'' signals). We take a similar approach and present a simple set of delay-insensitive gates \textsf{MERGE}, \textsf{FORK}, \textsf{DUAL}, and \textsf{CROSS}, as defined in \Cref{fig:custom-gates}.

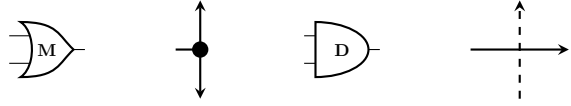
\begin{figure}[htb]
    \centering
    \begin{circuitikz}[scale=0.65, transform shape]
        \node[or port] (merge) at (0,0) {\textbf{M}};
        \draw[-, thick, >=stealth] (2, 0) -- (2.5, 0);
        \draw[->, thick, >=stealth] (2.5, 0) -- (2.5, 1);
        \draw[->, thick, >=stealth] (2.5, 0) -- (2.5, -1);
        \node[circle, fill=black](fork) at (2.5, 0) {};
        \node[and port] (dual) at (6,0) {\textbf{D}};
        \draw[->, thick, dashed, >=stealth] (9, -1) -- (9, 1);
        \draw[->, thick, >=stealth] (8, 0) -- (10, 0);
        
    \end{circuitikz}
    \caption{From left to right: \textsf{MERGE} (joins signals from two wires into one, at most one input wire is allowed to carry signal), \textsf{FORK} (splits a signal into two output copies), \textsf{DUAL} (produces an output signal if signals arrive on both inputs) and \textsf{CROSS} (weak wire crossing, simultaneous signals on both wires are prohibited).}
    \label{fig:custom-gates}
\end{figure}

\begin{restatable}{lemma}{customgates}\label{lem:custom-gates}
    The set of delay-insensitive gates $\{ \textup{\textsf{MERGE}, \textsf{FORK}, \textsf{DUAL}, \textsf{CROSS}} \}$ can compute any Boolean function using a feed-forward planar circuit with dual-rail encoding.
\end{restatable}

\begin{proof}
    Due to \cref{thm:planar} it is sufficient to show that we can construct a \textsf{NAND} gate and fan-out operation using a planar circuit. We employ dual-rail encoding and for fan-out \textsf{FORK} both inputs $\texttt{I}$ and $\neg \texttt{I}$ individually and then \textsf{CROSS} them using the fact that at most one of them contains the signal. The circuit for the dual-rail \textsf{NAND} is given in \Cref{fig:di-nand}.
\end{proof}

    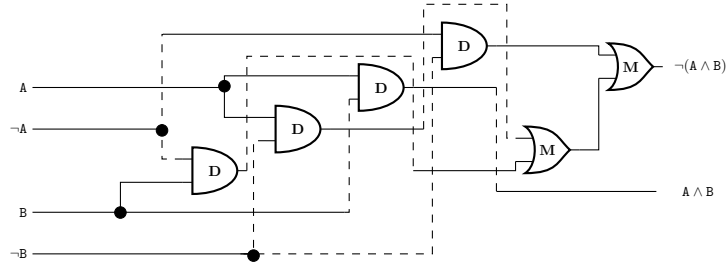
\begin{figure}[ht]
        \centering
        \begin{circuitikz}[scale=0.55, transform shape]
        \draw (2,4) node[and port] (myand1) {\textbf{D}}
              (4,5) node[and port] (myand2) {\textbf{D}}
              (6,6) node[and port] (myand3) {\textbf{D}}
              (8,7) node[and port] (myand4) {\textbf{D}}
              (10,4.5) node[or port] (myor1) {\textbf{M}}
              (12,6.5) node[or port] (myor2) {\textbf{M}};
    
        \draw (-3,6) node[left] (a) {$\texttt{A}$}
              (-3,5) node[left] (nota) {$\neg \texttt{A}$}
              (-3,3) node[left] (b) {$\texttt{B}$}
              (-3,2) node[left] (notb) {$\neg \texttt{B}$};
    
        \draw (12.3,6.5) node[right] (a) {$\neg (\texttt{A} \wedge \texttt{B})$}
              (12.5,3.5) node[right] (nota) {$\texttt{A} \wedge \texttt{B}$};
              
        \draw[] (myand2.out) -- ++(2.25, 0);
        \draw[dashed] (myand2.out) ++(2.25, 0) -- ++(0, 3) -- ++(2, 0) |- (myor1.in 1);
        \draw[dashed] (myand1.out) -- ++(0, 2.75) -- ++(4, 0) -- ++(0, -2.75);
        \draw (myand1.out) ++(4, 0)  -| (myor1.in 2);
        \draw (myand4.out) -| (myor2.in 1);
        \draw (myor1.out) -| (myor2.in 2);
    
        \draw (myand2.in 1) -- ++(-1,0) |- (myand3.in 1) coordinate[pos=0.3] (a);
        \draw[-*] (-3,6) -| (a);
        \draw[] (myand3.out) -- ++(2, 0);
        \draw[dashed] (myand3.out) ++(2, 0) -- ++(0, -2.5);
        \draw (myand3.out) ++(2, -2.5) -- (12, 3.5);
            
        \draw[dashed] (myand1.in 1) -- ++(-0.5,0) |- (0.15, 7.275) coordinate[pos=0.14] (b);
        \draw[-*] (-3,5) -| (b);
        \draw (0.15, 7.275) |- (myand4.in 1);

        \draw[] (myand1.in 2) -- ++(-1.5,0) |- (4.5, 3) coordinate[pos=0.4] (c);
        \draw[dashed] (4.5, 3) -| (myand3.in 2);
        \draw[-*] (-3,3) -| (c);
    
        \draw[dashed] (myand2.in 2) -- ++(-0.3,0) |- (2, 2) coordinate[pos=0.48] (d);
        \draw[dashed] (2, 2) -| (myand4.in 2);
        \draw[-*] (-3,2) -| (d);

        \end{circuitikz}
    \caption{A planar construction for dual-rail \textsf{NAND} using 2 \textsf{MERGE}, 4 \textsf{FORK}, 4 \textsf{DUAL} and 12 \textsf{CROSS} gates. Neither \textsf{MERGE} nor \textsf{CROSS} can have two particles arrive at the same time.}
    \label{fig:di-nand}
    \end{figure}

As a direct consequence it follows that it is sufficient to find a $3$-state ACA $\mathbb{X} = (\{0, 1, 2\}, N, f, \Z^2)$ with von Neumann neighborhood capable of reliably simulating all delay insensitive gates from \Cref{fig:custom-gates}. We start by describing a ``helper'' ACA $\mathbb{X}' = (\{0, 1, 2\}, N, f', \Z^2)$ embedding a FAN and capable of implementing wires, \textsf{FORK} and \textsf{DUAL}. We assume $0$ as a ``background'' cell, i.e., a cell that never changes its state and defines no direction indicators. For the remaining nodes we define a mapping $\psi$ similarly to \Cref{prop:mountain-valley}, but now the direction indicator between adjacent non-zero cells points left (or down) if the states are equal, and right (or up) if they are distinct:
 \begin{align*}
    &\psi(a, a) = a \longleftarrow a  &\psi(a, b) &= a \longrightarrow b \quad \text{for $a \neq b$}
\end{align*}
We use the following principle (similarly to \Cref{ex:mountain-valley}) to describe the transition function $f'$ of $\mathbb{X}'$: as soon as all direction indicators point to some cell $i \in \Z^2$ it becomes active and can ``flip'' its state from $1$ to $2$ or from $2$ to $1$ respectively. Observe that this state change also causes all adjacent direction indicators to flip. To visualize the underlying circuit, we interpret any active cell of $\mathbb{X}'$ as containing the ``signal'', and in this way the signal is always transmitted reciprocally to direction indicators, which provides us with a simple way of encoding wires into configurations of $\mathbb{X}'$.

A possible wire construction is shown in \Cref{fig:gate-encoding} (first), and it is not unique--by swapping all $2$s with $1$s, we get the same wire carrying the same signal, hence there are multiple equally valid ways of transporting signals (this is  crucial for implementing crossings). We can easily construct wires into arbitrary directions, for example:
\begin{equation*}
    1 \longleftarrow 1 \longleftarrow 1 \longleftarrow 1 \longleftarrow 1 \longleftarrow 1 \quad\quad \text{and} \quad\quad 2 \longrightarrow 1 \longrightarrow 2 \longrightarrow 1 \longrightarrow 2 \longrightarrow 1
\end{equation*}
describe the left-right wire and the right-left wire, respectively. In a similar fashion, we can implement \textsf{DUAL} by combining two signals into one and \textsf{FORK} by splitting a signal into two, as shown in \Cref{fig:gate-encoding} (second and third). However, ACA $\mathbb{X}'$ can not yet be considered universal because the implementation for \textsf{MERGE} and \textsf{CROSS} seems to be impossible. Hence we augment the transition function $f'$ by adding four additional transitions to obtain the final ACA $\mathbb{X}$ (we use $\raisebox{-3.5pt}{\CustomSmall{w}}$ for $0$, $\raisebox{-3.5pt}{\CustomSmall{g}}$ for $1$ and $\raisebox{-3.5pt}{\CustomSmall{b}}$ for $2$). 
\begin{equation*}
    \raisebox{-8.5pt}{\Custom{g}{g}{g}{g}{w}} \mapsto \raisebox{-3.5pt}{\CustomSmall{b}} \raisebox{-8.5pt}{\Custom{b}{b}{g}{g}{w}} \mapsto \raisebox{-3.5pt}{\CustomSmall{b}} \text{ for } \textsf{MERGE} \quad\quad \raisebox{-8.5pt}{\Custom{g}{g}{w}{b}{g}} \mapsto \raisebox{-3.5pt}{\CustomSmall{g}} \raisebox{-8.5pt}{\Custom{g}{b}{w}{b}{b}} \mapsto \raisebox{-3.5pt}{\CustomSmall{b}} \text{ for } \textsf{CROSS}
\end{equation*}
\begin{figure}[tb]
    \centering
    \begin{tikzpicture}[font=\Large\sffamily, scale=0.65, transform shape]
        \foreach \i in {0,...,5}{
            \draw[dotted] (0.75*\i, 0) -- (0.75*\i, 3.75);
            \draw[dotted] (0, 0.75*\i) -- (3.75, 0.75*\i);

        \node[] at (0.375,0.375) {$\text{0}$};
        \node[] at (1.125,0.375) {$\text{0}$};
        \node[] at (1.875,0.375) {$\text{0}$};
        \node[] at (2.625,0.375) {$\text{2}$};
        \node[] at (3.375,0.375) {$\text{0}$};

        \node[] at (0.375,1.125) {$\text{0}$};
        \node[] at (1.125,1.125) {$\text{0}$};
        \node[] at (1.875,1.125) {$\text{0}$};
        \node[] at (2.625,1.125) {$\text{1}$};
        \node[] at (3.375,1.125) {$\text{0}$};
    
        \node[] at (0.375,1.875) {$\text{0}$};
        \node[] at (1.125,1.875) {$\text{0}$};
        \node[] at (1.875,1.875) {$\text{0}$};
        \node[] at (2.625,1.875) {$\text{2}$};
        \node[] at (3.375,1.875) {$\text{0}$};
    
        \node[] at (0.375,2.625) {$\text{2}$};
        \node[circle, fill=lightgray!50] at (1.125,2.625) {$\text{1}$};
        \node[] at (1.875,2.625) {$\text{1}$};
        \node[] at (2.625,2.625) {$\text{1}$};
        \node[] at (3.375,2.625) {$\text{0}$};
    
        \node[] at (0.375,3.375) {$\text{0}$};
        \node[] at (1.125,3.375) {$\text{0}$};
        \node[] at (1.875,3.375) {$\text{0}$};
        \node[] at (2.625,3.375) {$\text{0}$};
        \node[] at (3.375,3.375) {$\text{0}$};

        \draw[thick] (0, 0.75*4) -- (3, 0.75*4);
        \draw[thick] (0, 0.75*3) -- (2.25, 0.75*3);
        \draw[thick] (0.75*4, 0) -- (0.75*4, 3);
        \draw[thick] (0.75*3, 0) -- (0.75*3, 2.25);
        \draw[->, color=darkred, thick, >=stealth] (0.375 + 2.25, 0.6) -- (0.375 + 2.25, 0.9);
        \draw[->, color=darkred, thick, >=stealth] (0.375 + 2.25, 2.1) -- (0.375 + 2.25, 2.4);
        \draw[->, color=darkred, thick, >=stealth] (0.375 + 2.25, 1.35) -- (0.375 + 2.25, 1.65);

        \draw[->, color=darkred, thick, >=stealth] (0.6, 0.375 + 2.25) -- (0.9, 0.375 + 2.25);
        \draw[<-, color=darkred, thick, >=stealth] (2.1, 0.375 + 2.25) -- (2.4, 0.375 + 2.25);
        \draw[<-, color=darkred, thick, >=stealth] (1.35, 0.375 + 2.25) -- (1.65, 0.375 + 2.25);
        
        }
    \end{tikzpicture}
    \hspace{0.2cm}
    \begin{tikzpicture}[font=\Large\sffamily, scale=0.65, transform shape]

    \foreach \i in {0,...,5}{
        \draw[dotted] (0.75*\i, 0) -- (0.75*\i, 3.75);
        \draw[dotted] (0, 0.75*\i) -- (3.75, 0.75*\i);

    \node[] at (0.375,0.375) {$\text{0}$};
    \node[] at (1.125,0.375) {$\text{0}$};
    \node[] at (1.875,0.375) {$\text{0}$};
    \node[] at (2.625,0.375) {$\text{1}$};
    \node[] at (3.375,0.375) {$\text{0}$};

    \node[] at (0.375,1.125) {$\text{0}$};
    \node[] at (1.125,1.125) {$\text{0}$};
    \node[] at (1.875,1.125) {$\text{0}$};
    \node[] at (2.625,1.125) {$\text{2}$};
    \node[] at (3.375,1.125) {$\text{0}$};

    \node[] at (0.375,1.875) {$\text{2}$};
    \node[circle, fill=lightgray!50] at (1.125,1.875) {$\text{1}$};
    \node[] at (1.875,1.875) {$\text{1}$};
    \node[] at (2.625,1.875) {$\text{1}$};
    \node[] at (3.375,1.875) {$\text{0}$};

    \node[] at (0.375,2.625) {$\text{0}$};
    \node[] at (1.125,2.625) {$\text{0}$};
    \node[] at (1.875,2.625) {$\text{0}$};
    \node[] at (2.625,2.625) {$\text{1}$};
    \node[] at (3.375,2.625) {$\text{0}$};

    \node[] at (0.375,3.375) {$\text{0}$};
    \node[] at (1.125,3.375) {$\text{0}$};
    \node[] at (1.875,3.375) {$\text{0}$};
    \node[] at (2.625,3.375) {$\text{1}$};
    \node[] at (3.375,3.375) {$\text{0}$};

    \draw[thick] (0, 0.75*3) -- (2.25, 0.75*3);
    \draw[thick] (0, 0.75*2) -- (2.25, 0.75*2);
    \draw[thick] (0.75*4, 0) -- (0.75*4, 3.75);
    \draw[thick] (0.75*3, 0) -- (0.75*3, 1.5);
    \draw[thick] (0.75*3, 2.25) -- (0.75*3, 3.75);
    
    \draw[->, color=darkred, thick, >=stealth] (0.375 + 2.25, 0.6) -- (0.375 + 2.25, 0.9);
    \draw[->, color=darkred, thick, >=stealth] (0.375 + 2.25, 1.35) -- (0.375 + 2.25, 1.65);
    \draw[<-, color=darkred, thick, >=stealth] (0.375 + 2.25, 2.1) -- (0.375 + 2.25, 2.4);
    \draw[<-, color=darkred, thick, >=stealth] (0.375 + 2.25, 2.85) -- (0.375 + 2.25, 3.15);

    \draw[->, color=darkred, thick, >=stealth] (0.6, 0.375 + 1.5) -- (0.9, 0.375 + 1.5);
    \draw[<-, color=darkred, thick, >=stealth] (2.1, 0.375 + 1.5) -- (2.4, 0.375 + 1.5);
    \draw[<-, color=darkred, thick, >=stealth] (1.35, 0.375 + 1.5) -- (1.65, 0.375 + 1.5);        
    
    }
    \end{tikzpicture}
    \hspace{0.2cm}
    \begin{tikzpicture}[font=\Large\sffamily, scale=0.65, transform shape]

    \foreach \i in {0,...,5}{
        \draw[dotted] (0.75*\i, 0) -- (0.75*\i, 3.75);
        \draw[dotted] (0, 0.75*\i) -- (3.75, 0.75*\i);

    \node[] at (0.375,0.375) {$\text{0}$};
    \node[] at (1.125,0.375) {$\text{0}$};
    \node[] at (1.875,0.375) {$\text{2}$};
    \node[] at (2.625,0.375) {$\text{0}$};
    \node[] at (3.375,0.375) {$\text{0}$};

    \node[] at (0.375,1.125) {$\text{0}$};
    \node[] at (1.125,1.125) {$\text{0}$};
    \node[circle, fill=lightgray!50] at (1.875,1.125) {$\text{1}$};
    \node[] at (2.625,1.125) {$\text{0}$};
    \node[] at (3.375,1.125) {$\text{0}$};

    \node[] at (0.375,1.875) {$\text{0}$};
    \node[] at (1.125,1.875) {$\text{0}$};
    \node[] at (1.875,1.875) {$\text{1}$};
    \node[] at (2.625,1.875) {$\text{0}$};
    \node[] at (3.375,1.875) {$\text{0}$};

    \node[] at (0.375,2.625) {$\text{2}$};
    \node[circle, fill=lightgray!50] at (1.125,2.625) {$\text{1}$};
    \node[] at (1.875,2.625) {$\text{1}$};
    \node[] at (2.625,2.625) {$\text{1}$};
    \node[] at (3.375,2.625) {$\text{1}$};

    \node[] at (0.375,3.375) {$\text{0}$};
    \node[] at (1.125,3.375) {$\text{0}$};
    \node[] at (1.875,3.375) {$\text{0}$};
    \node[] at (2.625,3.375) {$\text{0}$};
    \node[] at (3.375,3.375) {$\text{0}$};

    \draw[thick] (0, 0.75*4) -- (3.75, 0.75*4);
    \draw[thick] (0, 0.75*3) -- (1.5, 0.75*3);
    \draw[thick] (2.25, 0.75*3) -- (3.75, 0.75*3);
    \draw[thick] (0.75*2, 0) -- (0.75*2, 2.25);
    \draw[thick] (0.75*3, 0) -- (0.75*3, 2.25);
    \draw[->, color=darkred, thick, >=stealth] (0.375 + 1.5, 0.6) -- (0.375 + 1.5, 0.9);
    \draw[<-, color=darkred, thick, >=stealth] (0.375 + 1.5, 2.1) -- (0.375 + 1.5, 2.4);
    \draw[<-, color=darkred, thick, >=stealth] (0.375 + 1.5, 1.35) -- (0.375 + 1.5, 1.65);

    \draw[->, color=darkred, thick, >=stealth] (0.6, 0.375 + 2.25) -- (0.9, 0.375 + 2.25);
    \draw[<-, color=darkred, thick, >=stealth] (1.35, 0.375 + 2.25) -- (1.65, 0.375 + 2.25);
    \draw[<-, color=darkred, thick, >=stealth] (2.1, 0.375 + 2.25) -- (2.4, 0.375 + 2.25);
    \draw[<-, color=darkred, thick, >=stealth] (2.85, 0.375 + 2.25) -- (3.15, 0.375 + 2.25);
    }
    \end{tikzpicture}
    \hspace{0.2cm}
    \begin{tikzpicture}[font=\Large\sffamily, scale=0.65, transform shape]

    \foreach \i in {0,...,5}{
        \draw[dotted] (0.75*\i, 0) -- (0.75*\i, 3.75);
        \draw[dotted] (0, 0.75*\i) -- (3.75, 0.75*\i);

    \node[] at (0.375,0.375) {$\text{0}$};
    \node[] at (1.125,0.375) {$\text{0}$};
    \node[] at (1.875,0.375) {$\text{1}$};
    \node[] at (2.625,0.375) {$\text{0}$};
    \node[] at (3.375,0.375) {$\text{0}$};

    \node[] at (0.375,1.125) {$\text{0}$};
    \node[] at (1.125,1.125) {$\text{0}$};
    \node[] at (1.875,1.125) {$\text{1}$};
    \node[] at (2.625,1.125) {$\text{0}$};
    \node[] at (3.375,1.125) {$\text{0}$};

    \node[] at (0.375,1.875) {$\text{2}$};
    \node[] at (1.125,1.875) {$\text{2}$};
    \node[] at (1.875,1.875) {$\text{0}$};
    \node[] at (2.625,1.875) {$\text{2}$};
    \node[] at (3.375,1.875) {$\text{2}$};

    \node[] at (0.375,2.625) {$\text{0}$};
    \node[] at (1.125,2.625) {$\text{0}$};
    \node[] at (1.875,2.625) {$\text{1}$};
    \node[] at (2.625,2.625) {$\text{0}$};
    \node[] at (3.375,2.625) {$\text{0}$};

    \node[] at (0.375,3.375) {$\text{0}$};
    \node[] at (1.125,3.375) {$\text{0}$};
    \node[] at (1.875,3.375) {$\text{1}$};
    \node[] at (2.625,3.375) {$\text{0}$};
    \node[] at (3.375,3.375) {$\text{0}$};

    \draw[thick] (0, 0.75*3) -- (1.5, 0.75*3);
    \draw[thick] (0, 0.75*2) -- (1.5, 0.75*2);
    \draw[thick] (2.25, 0.75*3) -- (3.75, 0.75*3);
    \draw[thick] (2.25, 0.75*2) -- (3.75, 0.75*2);
    \draw[thick] (0.75*2, 0) -- (0.75*2, 1.5);
    \draw[thick] (0.75*2, 2.25) -- (0.75*2, 3.75);
    \draw[thick] (0.75*3, 0) -- (0.75*3, 1.5);
    \draw[thick] (0.75*3, 2.25) -- (0.75*3, 3.75);
    
    \draw[<-, color=darkred, thick, >=stealth] (0.375 + 1.5, 0.6) -- (0.375 + 1.5, 0.9);
    \draw[<-, color=darkred, thick, >=stealth] (0.375 + 1.5, 2.85) -- (0.375 + 1.5, 3.15);

    \draw[<-, color=darkred, thick, >=stealth] (0.6, 0.375 + 1.5) -- (0.9, 0.375 + 1.5);     
    \draw[<-, color=darkred, thick, >=stealth] (2.85, 0.375 + 1.5) -- (3.15, 0.375 + 1.5);
    
    }
    \end{tikzpicture}
    \hspace{0.2cm}
    \begin{tikzpicture}[font=\Large\sffamily, scale=0.65, transform shape]

    \foreach \i in {0,...,5}{
        \draw[dotted] (0.75*\i, 0) -- (0.75*\i, 3.75);
        \draw[dotted] (0, 0.75*\i) -- (3.75, 0.75*\i);

    \node[] at (0.375,0.375) {$\text{0}$};
    \node[] at (1.125,0.375) {$\text{0}$};
    \node[] at (1.875,0.375) {$\text{0}$};
    \node[] at (2.625,0.375) {$\text{0}$};
    \node[] at (3.375,0.375) {$\text{0}$};

    \node[] at (0.375,1.125) {$\text{1}$};
    \node[] at (1.125,1.125) {$\text{1}$};
    \node[] at (1.875,1.125) {$\text{1}$};
    \node[] at (2.625,1.125) {$\text{1}$};
    \node[] at (3.375,1.125) {$\text{1}$};

    \node[] at (0.375,1.875) {$\text{0}$};
    \node[] at (1.125,1.875) {$\text{0}$};
    \node[] at (1.875,1.875) {$\text{2}$};
    \node[] at (2.625,1.875) {$\text{0}$};
    \node[] at (3.375,1.875) {$\text{0}$};

    \node[] at (0.375,2.625) {$\text{0}$};
    \node[] at (1.125,2.625) {$\text{0}$};
    \node[] at (1.875,2.625) {$\text{1}$};
    \node[] at (2.625,2.625) {$\text{0}$};
    \node[] at (3.375,2.625) {$\text{0}$};

    \node[] at (0.375,3.375) {$\text{0}$};
    \node[] at (1.125,3.375) {$\text{0}$};
    \node[] at (1.875,3.375) {$\text{2}$};
    \node[] at (2.625,3.375) {$\text{0}$};
    \node[] at (3.375,3.375) {$\text{0}$};

    \draw[thick] (0, 0.75*1) -- (3.75, 0.75*1);
    \draw[thick] (0, 0.75*2) -- (1.5, 0.75*2);
    \draw[thick] (2.25, 0.75*2) -- (3.75, 0.75*2);
    \draw[thick] (0.75*2, 1.5) -- (0.75*2, 3.75);
    \draw[thick] (0.75*3, 1.5) -- (0.75*3, 3.75);
    \draw[->, color=darkred, thick, >=stealth] (0.375 + 1.5, 2.85) -- (0.375 + 1.5, 3.15);
    \draw[->, color=darkred, thick, >=stealth] (0.375 + 1.5, 2.1) -- (0.375 + 1.5, 2.4);
    \draw[->, color=darkred, thick, >=stealth] (0.375 + 1.5, 1.35) -- (0.375 + 1.5, 1.65);

    \draw[<-, color=darkred, thick, >=stealth] (0.6, 0.375 + 0.75) -- (0.9, 0.375 + 0.75);
    \draw[<-, color=darkred, thick, >=stealth] (1.35, 0.375 + 0.75) -- (1.65, 0.375 + 0.75);
    \draw[<-, color=darkred, thick, >=stealth] (2.1, 0.375 + 0.75) -- (2.4, 0.375 + 0.75);
    \draw[<-, color=darkred, thick, >=stealth] (2.85, 0.375 + 0.75) -- (3.15, 0.375 + 0.75);
    }
    \end{tikzpicture}
    \caption{A wire transmitting a signal (highlighted in gray) to the right and to the bottom (first). An implementation of \textsf{FORK} splitting a signal into two copies (second) and an implementation of \textsf{DUAL} asynchronously combining two signals into one (third). A \textsf{CROSS} gate flips $0$ to $1$ on a signal coming from the left and flips $0$ to $2$ on a signal coming from the bottom (forth). A \textsf{MERGE} flips the center cell from $1$ to $2$ if a signal arrives from the top or from the left (fifth).}
    \label{fig:gate-encoding}
\end{figure}
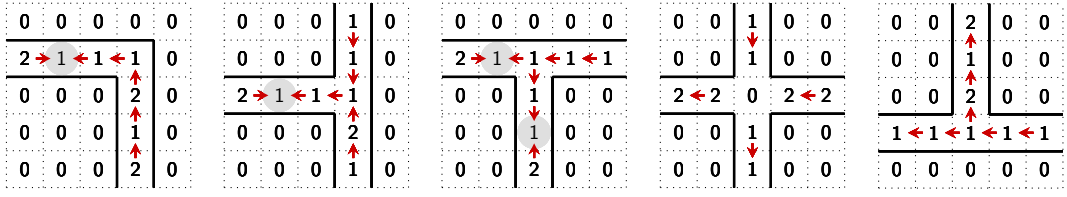
Notice that the new active transitions do not affect wires, \textsf{FORK} or \textsf{DUAL} from \Cref{fig:gate-encoding}. The new \textsf{MERGE} functions like a \textsf{DUAL}, but it is sufficient to flip just one direction indicator to transmit a signal, see \Cref{fig:gate-encoding} (fifth). Similarly, for \textsf{CROSS}, as soon as one signal arrives at the crossing, it gets immediately transmitted by ``flipping'' the $0$ cell and without creating any interference, as shown in \Cref{fig:gate-encoding} (forth). Multiple signals arriving simultaneously can cause potential problems, but this is forbidden by definitions of \textsf{MERGE} and \textsf{CROSS}. 

We note that this implementation of gates ``breaks'' them as soon as they have processed a signal, but this will not have any effect on feed-forward circuits, since those do not contain loops. It is clear that as long as wires have a spacing of at least two $0$ cells in between, any feed-forward circuit embedded in $\mathbb{X}$ operates correctly, hence is universal. We verified our construction using the \texttt{Golly} \cite{golly} simulator (implementation provided in \appendixref{sec:app-gates}).

\begin{remark}
    Notice that due to introduction of \textsf{MERGE} and \textsf{CROSS} gates, $\mathbb{X}$ no longer has a ``dual'' flip automata network as in \Cref{lem:causality}, therefore it does not guarantee invariant histories. Nevertheless, it still has invariant histories on all initial configurations which correspond to valid circuits.
\end{remark}



\bibliography{main}

\clearpage

\ifincludeappendix

\appendix

\section{Appendix: Universality}\label{sec:app-universality}

In \cref{sec:universality} we presented the concept of intrinsic universality, which is a stronger notion of universality in cellular automata: if some cellular automaton $A$ is capable of simulating any cellular automaton it is capable of simulating any Turing machine. In fact, Ollinger \cite{non-universal} showed that for almost any reasonable definition of universality on cellular automata, there will always be universal cellular automata, which are not intrinsically universal. This means that we require a separate, ``weaker'' notion of universality for (a)synchronous cellular automata. 

To define a computation on cellular automata we require, similarly to (universal) Turing machines \cite{small-utms}, the notions of encoder (to initialize the system) and decoder (to read the output). Moreover, since cellular automata are non-terminating by design, one needs to define a \emph{halting condition} which marks the point when the output can be read from the configuration of a cellular automaton (or from its orbit). One has to be exceptionally careful when designing these components:
\begin{itemize}
    \item Contrary to Turing machines there is no notion for a ``blank'' cell in cellular automata (every cell is independently capable of local computation), hence all inputs are infinite by design
    \item Since configurations are infinite, one has to limit their complexity (otherwise, enumerating all possible outputs in the initial configuration makes identity automata universal)
    \item Neither encoder nor decoder is allowed to be too powerful, otherwise this can be abused to ``outsource'' the computation to encoding or decoding subroutines instead of the automaton itself
\end{itemize}
Most historical universality constructions for cellular automata were \emph{finitistic} in nature and followed the approach of using only $q$-\emph{finite} configurations for computation, i.e., requiring initial conditions to have all but a finite number of cells in some specified state $q$ \cite{durand-gol}. However, after Banks \cite{Banks} constructed a rotation-symmetric $2$-state von Neumann intrinsically universal cellular automaton and, Codd has proven that no automaton with such properties can be universal on $q$-finite configurations, a problem with this approach became apparent \cite{lindgren}. Because of that, modern definitions of universality are no longer finitistic: after all, what \emph{is} ``void'' in the context of cellular automata?

Ultimately periodic configurations from \cref{def:ultimately-periodic} are a natural choice for initial conditions--while being infinite, they carry only a finite amount of information. Moreover, our slightly adapted version allows to naturally embed $1$-dimensional ultimately periodic configurations into $2$-dimensional ones. Indeed, this is extremely desirable, since we expect any natural extension of universal $d$-dimensional cellular automaton--for example by allowing it to shift across the new dimension, or enforce certain boundary conditions--to remain universal in a $(d+1)$-dimensional setting.

Finally, it makes sense to require halting behavior to be so simple that it can be easily detected by a local observer, without performing any \emph{actual} computation. The above considerations lead to the following definition of universality in cellular automata, as presented in \cite{durand-gol}.

\begin{definition}[Universality in Cellular Automata \cite{durand-gol}]\label{def:universal-ca}
     A cellular automaton $A = (S, N, f, \Z^d)$ is called universal if for some partial computable universal function $U \colon {\{0, 1\}}^\ast \longrightarrow {\{0, 1\}}^\ast \cup \{\bot\}$ we can find an ultimately periodic configuration $c^\infty$ (compiler) and total computable functions $e \colon {\{0, 1 \}}^\ast \longrightarrow \bigl\lceil S^{\Z^d} \bigl\rceil$ (encoder) and $d \colon \bigl\lceil S^{\Z^d} \bigl\rceil \longrightarrow {\{0, 1 \}}^\ast$ (decoder), a pattern $ p\colon D \longrightarrow S$ with $D \subset \Z^{d}$ finite (halting condition) such that for every input $w \in {\{0, 1\}}^\ast$:
     \begin{enumerate}
         \item The initial configuration $c_0 = \widetilde{e(w)}$ is a finite perturbation of $c^\infty$, i.e. $\| c_0 - c^\infty \| < \infty$ where  $\| \cdot \|$ is the $l_0$ norm on  $S^{\Z^d}$.
         
        \item If $U(w)$ is defined there exists the earliest time stamp $t \in \N$ in the orbit $c_0, c_1, \ldots$ such that $c_t$ contains $p$, and satisfies $d(\lceil c_t \rceil) = U(w)$.

        \item If $U(w) = \bot$ then $p$ will never appear in the orbit $c_0, c_1, \ldots$ 
     \end{enumerate}
\end{definition}

\Cref{def:universal-ca} has proven itself to be both sufficiently general and sufficiently robust: almost all universal constructions that we have encountered in literature satisfy it. Durand \cite{durand-gol} also explains why many less restrictive modifications of \Cref{def:universal-ca} allow universality of trivial automata (some earlier attempts lead to ``shift'' automata being wrongly misclassified as universal) and why more restrictive definitions might miss important special cases.

\section{Appendix: Examples}\label{app:examples}

\begin{example}[Verifying Commutativity]\label{example:local-commutativity}
    Given an ACA $\mathbb{A} = (S, N, f, \mathbb{Z}^d)$ with the global transition function $G$ observe that for any singleton update set $\{i\}$, the transformation $G_{\{i\}}$ can be fully captured by the local rule function $f$. Hence, given two singleton update sets $\{i\}$ and $\{j\}$, to verify commutativity of $G$, i.e. whether
    \begin{equation*}
        G_{\{i\}} \circ G_{\{j\}} = G_{\{j\}} \circ G_{\{i\}}
    \end{equation*}
    one only needs to check the cases where cells $i$ and $j$ are neighboring, i.e. for $N = (n_1, \ldots, n_k)$
    \begin{equation*}
        i \in \{j+n_1, \ldots, j+n_k \}
    \end{equation*}
    since in a different setting, the update of $i$ will not affect the update of $j$. Thus, we can verify commutativity by considering all possible interactions of the local rule function with itself. In the case of a first neighbors cellular automaton, it is sufficient to show that for any states $s_0, s_1, s_2, s_3 \in S$ with $f(s_0, s_1, s_2) \neq s_1$ and $f(s_1, s_2, s_3) \neq s_2$ (we call such transitions \emph{active}) it holds:
    \begin{align*}
        & f(f(s_0, s_1, s_2), s_2, s_3) = f(s_1, s_2, s_3) & f(s_0, s_1, f(s_1, s_2, s_3)) = f(s_0, s_1, s_2)
    \end{align*}
    Since neighborhoods of size $N$ have at most $N-1$ overlaps (see \Cref{fig:overlaps} for an example), we conclude that commutativity can be verified in $\mathcal{O}(|N|{|f|}^2)$ time, where $|f| \leq S^{|N|}$ denotes the number of active transitions in $f$.

    \begin{figure}[htb]
    \centering
    \begin{tikzpicture}[scale=0.8, transform shape]
        \draw[thick, fill=lightgray] (0.5, 0.5) rectangle (1.5, 1.5);
        \draw[thick] (0, 0) rectangle (1.5, 1.5);
        \draw[thick] (0, 0.5) -- (1.5, 0.5);
        \draw[thick] (0, 1) -- (1.5, 1);
        \draw[thick] (0.5, 0) -- (0.5, 1.5);
        \draw[thick] (1, 0) -- (1, 1.5);
        \draw[thick] (0.5, 0.5) rectangle (2.0, 2.0);
        \draw[thick] (0.5, 1.0) -- (2, 1);
        \draw[thick] (0.5, 1.5) -- (2, 1.5);
        \draw[thick] (1, 0.5) -- (1, 2);
        \draw[thick] (1.5, 0.5) -- (1.5, 2.0);
        \draw[thick, fill=lightgray] (4, 0.5) rectangle (5, 1.5);
        \draw[thick] (4.0, 0) rectangle (5.5, 1.5);
        \draw[thick] (4.0, 0.5) -- (5.5, 0.5);
        \draw[thick] (4.0, 1) -- (5.5, 1);
        \draw[thick] (4.5, 0) -- (4.5, 1.5);
        \draw[thick] (5, 0) -- (5, 1.5);
        \draw[thick] (3.5, 0.5) rectangle (5.0, 2.0);
        \draw[thick] (3.5, 1.0) -- (5, 1);
        \draw[thick] (3.5, 1.5) -- (5, 1.5);
        \draw[thick] (4, 0.5) -- (4, 2);
        \draw[thick] (4.5, 0.5) -- (4.5, 2.0);
        \draw[thick, fill=lightgray] (7.5, 0) rectangle (8.5, 1.5);
        \draw[thick] (7.0, 0) rectangle (8.5, 1.5);
        \draw[thick] (7.0, 0.5) -- (8.5, 0.5);
        \draw[thick] (7.0, 1) -- (8.5, 1);
        \draw[thick] (7.5, 0) -- (7.5, 1.5);
        \draw[thick] (8, 0) -- (8, 1.5);
        \draw[thick] (7.5, 0) rectangle (9, 1.5);
        \draw[thick] (7.5, 0.5) -- (9, 0.5);
        \draw[thick] (7.5, 1) -- (9, 1);
        \draw[thick] (8, 0) -- (8, 1.5);
        \draw[thick] (8.5, 0) -- (8.5, 1.5);
        \draw[thick, fill=lightgray] (10.5, 0.5) rectangle (12, 1.5);
        \draw[thick] (10.5, 0) rectangle (12, 1.5);
        \draw[thick] (10.5, 0.5) -- (12, 0.5);
        \draw[thick] (10.5, 1) -- (12, 1);
        \draw[thick] (11, 0) -- (11, 1.5);
        \draw[thick] (11.5, 0) -- (11.5, 1.5);
        \draw[thick] (10.5, 0.5) rectangle (12, 2);
        \draw[thick] (10.5, 1) -- (12, 1);
        \draw[thick] (10.5, 1.5) -- (12, 1.5);
        \draw[thick] (11, 0.5) -- (11, 2);
        \draw[thick] (11.5, 0.5) -- (11.5, 2);
        \end{tikzpicture}
        \caption{Eight possible overlaps of the local transition function $f$ with itself for Moore's neighborhood (each overlap can appear twice).}
        \label{fig:overlaps}
    \end{figure}
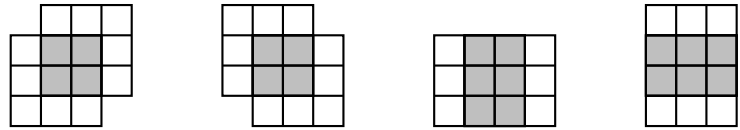
\end{example}

\begin{example}[Marching Soldiers \cite{gacs, nakamura}]\label{ex:marching-soldiers}
    Here we present a classical technique from \cite{gacs, nakamura} which allows us to simulate a CA with $q$ states using an ACA with $3q^2$ states in real time, and show that it abides by \Cref{def:invariant}.

    Let $A = (S, N, f, \Z^d)$ be a CA with symmetric neighborhood $N$, then we define an ACA $\mathbb{B} = (S \times S \times \{1, 2, 3\}, N, f', \Z^d)$ where each state in $\mathbb{B}$ is a triplet consisting of the ``current'' state $s^\text{cur}$, ``previous'' state $s^\text{prev}$ and a timer $t$. Let us denote with $\oplus$ the modulo $3$ addition in $\{1, 2, 3\}$, and for brevity only consider the case where $N = (-1, 0, 1)$ is first neighbors. Then, if $t_{-1}, t_1 \in \{t_0, t_0 \oplus 1\}$ the current cell becomes active and we can define the update function $f'$ as
    \begin{equation*}
        f'((s_{-1}^\text{prev}, s_{-1}^\text{cur}, t_{-1}), (s_0^\text{prev}, s_0^\text{cur}, t_0), (s_{1}^\text{prev}, s_{1}^\text{cur}, t_{1})) = (s_0^\text{cur}, f( \tau , s_0^\text{cur}, \eta ), t \oplus 1)
    \end{equation*}
    where $\tau = s_{-1}^\text{cur}$ if $t_{-1} = t_0$ and $\tau = s_{-1}^\text{prev}$ if $t_{-1} = t_0 \oplus 1$ (and analogously for $\eta$ and $s_1^\text{cur}, s_1^\text{prev}$). In other words, a cell is active if its timer is either one time step behind or at the same level in the update history as all of its neighbors, in which case it can use its $s^\text{cur}$ and either the $s^\text{prev}$ (if neighbor's timer is ahead) or $s^\text{cur}$ (if neighbor's timer is equal) state of its neighbor as an input to simulate transition function $f$ and update its state and timer. We emphasize that the current $s^\text{cur}$ becomes the next $s^\text{prev}$, as shown in Figure \ref{fig:XOR-marching}.

    \begin{figure}[H]
    \centering
    \begin{tikzpicture}[font=\sffamily\scriptsize]
    
    \draw[->, thick] (-0.5,0.0) -- (5.5,0.0) node[below] {$\mathbb{Z}$};
    \draw[->, thick] (-0.5,0.0) -- (-0.5,4) node[left] {$t$};

    \draw[thick, fill = lightgray, fill opacity = 0.3] (-0.5, 0.0) -- (-0.5,1.0) -- (0.0, 1.0) -- (0.0, 1.5) -- (0.5, 1.5) -- (0.5, 2) -- (1, 2) -- (1, 2.5) --  (3.5, 2.5) -- (3.5, 2) -- (4, 2) -- (4, 1.5) -- (4.5, 1.5) --(4.5, 1) --(5, 1) -- (5, 0.0) -- cycle;

    \draw[] (0, 0) -- (0, 1.0);
    \draw[] (0.5, 0) -- (0.5, 1.5);
    \draw[] (1, 0) -- (1, 2);
    \draw[] (1.5, 0) -- (1.5, 2.5);
    \draw[] (2, 0) -- (2, 2.5);
    \draw[] (2.5, 0) -- (2.5, 2.5);
    \draw[] (3, 0) -- (3, 2.5);
    \draw[] (3.5, 0) -- (3.5, 2.5);
    \draw[] (4, 0) -- (4, 2);
    \draw[] (4.5, 0) -- (4.5, 1.5);

    \draw[] (-0.5,0.5) -- (5, 0.5);
    \draw[] (0,1) -- (4.5, 1);
    \draw[] (0.5,1.5) -- (4, 1.5);
    \draw[] (1,2) -- (3.5, 2);
    \draw[] (1.5,2.5) -- (3, 2.5);

    \draw[dashed] (-0.5,0.5) -- (5, 0.5);
    \draw[dashed] (-0.5,1) -- (5, 1);
    \draw[dashed] (-0.5,1.5) -- (5, 1.5);
    \draw[dashed] (-0.5,2) -- (5, 2);
    \draw[dashed] (-0.5,2.5) -- (5, 2.5);

    \node[] at (-0.25,0.75) {$\text{0|0}$};
    \node[] at (0.25,0.75) {$\text{1|1}$};
    \node[] at (0.75,0.75) {$\text{0|0}$};
    \node[] at (1.25,0.75) {$\text{0|0}$};
    \node[] at (1.75,0.75) {$\text{1|1}$};
    \node[] at (2.25,0.75) {$\text{1|1}$};
    \node[] at (2.75,0.75) {$\text{0|0}$};
    \node[] at (3.25,0.75) {$\text{1|1}$};
    \node[] at (3.75,0.75) {$\text{1|1}$};
    \node[] at (4.25,0.75) {$\text{1|1}$};
    \node[] at (4.75,0.75) {$\text{0|0}$};

    \node[] at (0.25,1.25) {$\text{1|1}$};
    \node[] at (0.75,1.25) {$\text{1|0}$};
    \node[] at (1.25,1.25) {$\text{1|0}$};
    \node[] at (1.75,1.25) {$\text{0|1}$};
    \node[] at (2.25,1.25) {$\text{0|1}$};
    \node[] at (2.75,1.25) {$\text{0|0}$};
    \node[] at (3.25,1.25) {$\text{1|1}$};
    \node[] at (3.75,1.25) {$\text{1|1}$};
    \node[] at (4.25,1.25) {$\text{0|1}$};

    \node[] at (0.75,1.75) {$\text{1|1}$};
    \node[] at (1.25,1.75) {$\text{0|1}$};
    \node[] at (1.75,1.75) {$\text{1|0}$};
    \node[] at (2.25,1.75) {$\text{0|0}$};
    \node[] at (2.75,1.75) {$\text{0|0}$};
    \node[] at (3.25,1.75) {$\text{0|0}$};
    \node[] at (3.75,1.75) {$\text{0|1}$};

    \node[] at (1.25,2.25) {$\text{0|0}$};
    \node[] at (1.75,2.25) {$\text{1|1}$};
    \node[] at (2.25,2.25) {$\text{1|0}$};
    \node[] at (2.75,2.25) {$\text{0|0}$};
    \node[] at (3.25,2.25) {$\text{0|0}$};

    \node[] at (5.25,0.25) {$\textbf{1}$};
    \node[] at (5.25,0.75) {$\textbf{2}$};
    \node[] at (5.25,1.25) {$\textbf{3}$};
    \node[] at (5.25,1.75) {$\textbf{1}$};
    \node[] at (5.25,2.25) {$\textbf{2}$};
    \node[] at (5.25,2.75) {$\textbf{Timer}$};

    \node[] at (2.25, 2.75) {\textbf{active}};
    
    \end{tikzpicture}
    \caption{An asynchronous real time simulation of \rle{150} (\textsf{XOR}) using marching soldiers with $12$ states. All cells on the same level have the same timer value, and each tuple represents $s^{\text{cur}} | s^\text{prev}$. Notice that only the $3$ central cells are currently active.}
    \label{fig:XOR-marching}
    \end{figure}
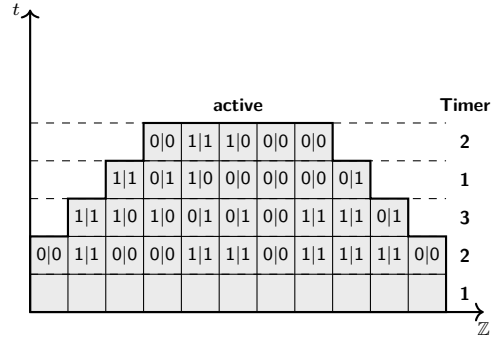

    Intuitively, since by design the ``time gap'' between neighboring cells can be at most one time step, every cell always has enough information to simulate a transition. To see that $\mathbb{B}$ simulates $A$, we notice that $\mathbb{B}$ is commutative \cite{gacs}, hence by \cref{thm:commutativity} has invariant histories. Moreover, if we initialize a starting configuration with timer value $t = 1$ everywhere, the (invariant) history of $\mathbb{B}$ corresponds exactly to the synchronous space-time of $\mathbb{B}$. Now we define a mapping $\psi: (s^{\text{prev}}, s^{\text{cur}}, t) \mapsto s^{\text{cur}}$ and observe that $G^A \prec H^B$.

\end{example}

\begin{example}[One-Way Simulation \cite{morita-book}]\label{ex:one-way}
    For any first neighbors cellular automaton $A = (S, N, f, \Z)$ we can construct a one-way cellular automaton $B = (S \times S, (-1, 0), f', \Z)$ that uses one cell of $B$ to simulate two cells of $A$. For any $s_{-3}, s_{-2}, s_{-1}, s_0 \in S$ we define
    \begin{equation*}
        f'(s_{-3}s_{-2}, s_{-1}s_0) = f(s_{-3}, s_{-2}, s_{-1}) f(s_{-2}, s_{-1}, s_0)
    \end{equation*}
    where we use $ab$ notation to describe a tuple containing $a$ and $b$. Observe that this way the simulated configuration is constantly shifting to the right (see \Cref{fig:one-way-simulation}).
    \begin{figure}[H]
        \centering
        \begin{tikzpicture}
        \tikzstyle{every path}=[semithick]
        \tikzstyle{every node}=[font=\sffamily\footnotesize]
        \tikzstyle{tmtape}=[draw,minimum size=0.6cm]
        
        \begin{scope}[start chain=1 going right,node distance=-0.15mm]
            \node [on chain=1,tmtape] {$\scriptstyle \ldots$};
            \node [on chain=1,tmtape] {1 \quad 0};
            \node [on chain=1,tmtape] {0 \quad 0};
            \node [on chain=1,tmtape, fill=lightgray!30] (l) {1 \quad \textcolor{darkred}{\textbf{0}}};
            \node [on chain=1,tmtape, fill = lightgray!30] (r){\textcolor{darkred}{\textbf{0}} \quad \textcolor{darkred}{\textbf{1}}};
            \node [on chain=1,tmtape] {1 \quad 1};
            \node [on chain=1,tmtape] {1 \quad 0};
            \node [on chain=1,tmtape] (last1) {$\scriptstyle \ldots$};
            \node [right = 0.5cm of last1] {$t=0$};
        \end{scope}
        
        \begin{scope}[yshift=-1cm,start chain=1 going right,node distance=-0.15mm]
            \node [on chain=1,tmtape] {$\scriptstyle \ldots$};
            \node [on chain=1,tmtape] {0 \quad 0};
            \node [on chain=1,tmtape] {1 \quad 0};
            \node [on chain=1,tmtape] {1 \quad 1};
            \node [on chain=1,tmtape, fill = lightgray!30] (n){1 \quad \textcolor{darkred}{\textbf{1}}};
            \node [on chain=1,tmtape] {0 \quad 1};
            \node [on chain=1,tmtape] {1 \quad 0};
            \node [on chain=1,tmtape] (last2){$\scriptstyle \ldots$};
            \node [right = 0.5cm of last2] {$t=1$};
        \end{scope}
        \end{tikzpicture}
        \caption{Simulating first neighbors \rle{150} (\textsf{XOR}) automaton using a one-way automaton $B$ by storing two cells inside a single cell. The domain of dependency for a selected cell in $B$ is highlighted in gray and for a simulated cell in red.}
        \label{fig:one-way-simulation}
    \end{figure}
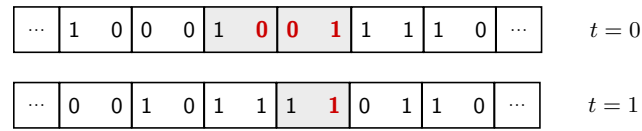
\end{example}

\section{Appendix: Omitted Proofs}\label{app:ommited-proofs}

\fliptheorem*

\begin{proof}
    We start by observing that in flip automata networks updating any subset of nodes at once is equivalent to updating all of these nodes one by one. Indeed, since none of the nodes in the update set are neighboring, updating one will not influence any others, hence any update schedule behaves in the same way as some sequential update schedule.
    
    Let $c_0$ be some initial configuration and consider two finite prefixes of update schedules $\zeta$ and $\zeta'$, where we only keep nodes that have all indicators pointing towards them (others do not affect the computation). Let $c$ and $c'$ be configurations after updating $\zeta$ and $\zeta'$ respectively, and let $\Ima$ be the function extracting the multiset of updated nodes from an update schedule (ignoring the order). 
    \begin{claim*}
        If $\Ima \zeta = \Ima \zeta'$ then $c(i) = c'(i)$ for any node $i$. 
    \end{claim*}
    Indeed, observe that both $c(i)$ and $c'(i)$ are only dependent on the finite set of updates $\mathcal{I}$ in the neighboring nodes. Let $\zeta_\mathcal{I}$ be the restriction of $\zeta$ on $\mathcal{I}$, and $\zeta'_\mathcal{I}$ be the restriction of $\zeta'$ respectively. Since $\Ima \zeta = \Ima \zeta'$ we clearly have $\Ima \zeta_\mathcal{I} = \Ima \zeta'_\mathcal{I}$. Due to previous observation, we can assume that all updates in $\zeta_\mathcal{I}$ and $\zeta'_\mathcal{I}$ are sequential. Thus $\zeta_\mathcal{I}$ and $\zeta'_\mathcal{I}$ are both finite and only differ in the ordering of individual updates. Now let $t$ be the earliest point where $\zeta_\mathcal{I}(t) \neq \zeta'_\mathcal{I}(t)$. We can find $\zeta_\mathcal{I}(t)$ within ${\zeta'_{\mathcal{I}}}(t), {\zeta'_{\mathcal{I}}}(t+1) \ldots $ (at some point $t'$) and simply insert it before $\zeta'_\mathcal{I}(t)$. Notice this will not have any effect on the later states in $\zeta'_\mathcal{I}$, since, due to the flip property, none of the neighbors of $\zeta'_\mathcal{I}(t')$ can appear between $t$ and $t'$. We can repeat this process until $\zeta_\mathcal{I} = \zeta'_\mathcal{I}$, concluding that both update schedules result in the same value, and hence $c(i) = c'(i)$. Finally, using the claim we can conclude that all flip automata networks have invariant histories since the final configuration is only dependent on the multiset of updated nodes and not their order.
\end{proof} 

\causality*

\begin{proof}
    To keep the proof concise, we present the construction for the $2$-dimensional von Neumann neighborhood only, but it is easy to generalize this outline to any dimension. The idea is to extend our ACA $A$ to a dual FAN $\mathcal{A}$ by assigning each (bidirectional) edge on the grid a direction indicator. Let $N=((1, 0), (0, 1), (0, 0), (0, -1), (-1, 0))$ be an explicit enumeration for von Neumann neighborhood and, and for arbitrary cell $i$ let $s_{(1, 0)}, s_{(0, 1)}, s, s_{(0, -1)}, s_{(-1, 0)}$ refer to the states in the neighborhood of $i$. Let us consider the neighbors with coordinates $(1, 0)$ and $(0, 1)$ and notice that due to the symmetry of the neighborhood, if for any cell $i$, we assign direction indicators $\in \{ \nearrow, \swarrow \}$ to its $(1, 0)$ and $(0, 1)$ neighbors, it will uniquely determine all direction indicators in the underlying FAN (is assigning indicators to top and right neighbors is sufficient to get a bijection between configurations). 
    
    We construct a directed tripartite graph $G_3 = (S \uplus S^E \uplus S^N, E)$ (see \Cref{fig:matching}) where $S^E$ and $S^N$ are copies of $S$ and we interpret the $S^E$ component as the ``eastern neighbor'' of $i$ (call the component \textbf{E}), and the $S^N$ component as the ``northern neighbor'' of $i$ (call the component \textbf{N}). Now let $s^N_{(k, l)}, s^E_{(k, l)}$ be the copies of $s_{(k, l)}$ within the \textbf{N} and \textbf{E} components respectively. The orientation of edges $E$ in $G_3$ between these components will determine the direction indicators within the FAN $\mathcal{A}$. 
    
    Consider the set of neighborhoods of all active cells (we call such neighborhoods/transitions \emph{active}), i.e. those which satisfy
    \begin{equation*}
        f(s_{(1, 0)}, s_{(0, 1)}, s, s_{(0, -1)}, s_{(-1, 0)}) \neq s
    \end{equation*} 
    Now, for every active transition in $f$
    \begin{itemize}
        \item We add an edge to $G_3$ oriented from $s^N_{(1, 0)}$ to $s_{(0, 0)}$ and from $s^E_{(0, 1)}$ to $s_{(0, 0)}$
        \item We add an edge to $G_3$ oriented from $s_{(0, 0)}$ to $s^N_{(-1, 0)}$ and from $s_{(0, 0)}$ to $s^E_{(0, -1)}$
    \end{itemize}
    \Cref{fig:matching} illustrates $G_3$ for a simple binary ACA $\mathbb{A}_{01} = (\{0, 1\}, N, f, \Z^2)$ with only two active transitions:
    \begin{align*}
        & f(0, 0, 1, 1, 0) = 0 & f(0, 1, 1, 0, 0) = 0
    \end{align*}
    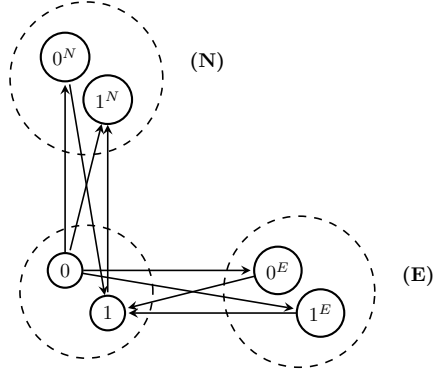
\begin{figure}[t]
    \centering
    \begin{tikzpicture}[
            > = stealth, 
            shorten > = 1pt, 
            auto,
            node distance = 3cm, 
            semithick, 
            scale=0.75, transform shape
        ]

        \tikzstyle{every state}=[
            draw = black,
            thick,
            fill = white,
            minimum size = 4mm
        ]

        \node[state] (s) {$0$};
        \node[state] (s1) [below right=0.3cm and 0.3cm of s] {$1$};
        \node[state] (s2) [above=3cm of s] {$0^N$};
        \node[state] (s3) [above=3cm of s1] {$1^N$};
        \node[state] (s4) [right=3cm of s] {$0^E$};
        \node[state] (s5) [right=3cm of s1] {$1^E$};
        
        \path[<-] (s1) edge node{} (s2);
        \path[<-] (s3) edge node{} (s1);
        \path[->] (s) edge node{} (s2);
        \path[<-] (s3) edge node{} (s);
        \path[->] (s) edge node{} (s4);
        \path[->] (s4) edge node{} (s1);
        \path[->] (s5) edge node{} (s1);
        \path[->] (s) edge node{} (s5);

        \node[above right = 0.05cm and 1cm of s3] {\textbf{(N)}};
        \node[above right = 0.05cm and 1cm of s5] {\textbf{(E)}};

        \node[circle, draw, dashed, inner xsep=1mm, inner ysep = 1mm, fit=(s) (s1)] {};
        \node[circle, draw, dashed, inner xsep=1mm, inner ysep = 1mm, fit=(s2) (s3)] {};
        \node[circle, draw, dashed, inner xsep=1mm, inner ysep = 1mm, fit=(s4) (s5)] {};
    \end{tikzpicture}
    \caption{Matching construction $G_3$ for ACA $\mathbb{A}_{01}$ with von Neumann neighborhood in two dimensions.}
    \label{fig:matching}
    \end{figure}
    After we have finished this process, we check if there are any nodes in the $S$ component that do not have an edge to some node in $S^N$ or $S^E$, in which case we add an edge to all missing nodes (oriented arbitrarily). Now, for any configuration $c \in S^{\Z^2}$ we assign direction indicators using the edge orientations from $G_3$ (and using the north/east rule) and observe the following:
    \begin{enumerate}
        \item There are no bidirectional edges in $G_3$ (otherwise, there must exist two transitions which are both active and which overlap).
        \item There is a directed edge between every node in $S$ and every node in $S^N$ and $S^E$ (enforced by construction).
        \item All edges on the grid $\Z^2$ have a unique direction indicator, and if a transition of cell $i$ changes its state, then all direction indicators from nodes
        \begin{equation*}
            \{ i + (1, 0), i + (0, 1), i + (0, -1), i + (-1, 0)\}
        \end{equation*}
        necessarily point to $i$ (enforced by north/east orientation).
    \end{enumerate}
    This means that the tuple $(S, f)$ induces direction indicators for any configuration $c$ on $\Z^2$ in such a manner that $\mathcal{A}$ becomes a valid flip automata network. FAN $\mathcal{A}$ relies on $G_3$ to define its local rule function $f'$ and determine which direction indicators will get flipped on each update. It follows that arbitrary evolutions of $\mathbb{A}$ naturally commute under this projection with evolutions of $\mathcal{A}$.
    
\end{proof}

\semiautomaton*

\begin{proof}
    We prove the statement by providing a simulation for $A$. Let $\mathbb{B}$ be as defined in the statement, then for any input string $w \in \Sigma^\ast$ with $|w| = n$ we define an ultimately periodic configuration $c_0$
    \begin{align*}
        &c_0(i) = w(i) \quad \text{for $i \in \{1, \ldots, n\}$} &c_0(i) = q_0\varnothing \quad \text{otherwise}
    \end{align*}
    Furthermore, we define the local rule function $f$ to mimic the transition function $\delta$ for all $s' \in \Sigma$ and $qs \in Q \times (\Sigma \cup \{\varnothing\})$ as $f(q s, s') = \delta(q, s')$ and in case $s' \in Q \times (\Sigma \cup \{\varnothing\})$ the state is never changed. We notice that at each point in time, independently of the asynchronous update schedule $\zeta$ at most one cell in $\mathbb{B}$ is active, hence $\mathbb{B}$ has invariant histories on $c_0$. Furthermore, by construction, it is clear that $\mathbb{B}$ reaches a limit point $c$ as soon as cells become tuple states, and $w$ gets overwritten by the output of $A$ such that $c$ now contains the output of $A$ on $w$ and $c(n)$ encodes the final state.
\end{proof}

\embedding*

\begin{proof}
    We give a proof sketch in the spirit of \cite{circuit-unviersality}. The main idea is to embed the synchronous space-time of arbitrary $1$-dimensional cellular automaton $A = (S, (-1, 0, 1), f, \Z)$ into every configuration of the specified $2$-dimensional automaton $B$. We uniquely encode the states in $s \in S$ as binary strings
    \begin{equation*}
        s \mapsto \lceil s \rceil \in {\{0, 1\}}^\ast
    \end{equation*}
    and design a Boolean function $\hat{f}$ such that
    \begin{equation*}
        \hat{f}(\lceil s_{-1} \rceil, \lceil s_{0} \rceil, \lceil s_{1} \rceil) = \lceil f(s_{-1}, s_0, s_1) \rceil
    \end{equation*}
    Since $B$ can implement any Boolean function we can construct a pattern $p$ embedding $\hat{f}$ and then periodically duplicate $p$ as shown in \Cref{fig:logical-universality}, such that individual $\hat{f}$-blocks are ``wired'' together to enable state sharing. It we embed initial configuration of $A$ into some row of initial configuration $c_0$ for $B$, it follows that the orbit of $c_0$ in $B$ will contain the embedding for the synchronous space-time evolution of $A$. Notice that if $c_0$ is ultimately periodic, then is the initial configuration of $B$ due to \Cref{def:ultimately-periodic}.
    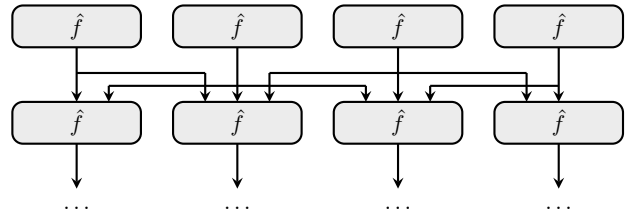
\begin{figure}[htb]
        \centering
        \begin{tikzpicture}[scale=0.85, transform shape]
            \draw[thick, rounded corners, fill=lightgray!30] (0, 3) rectangle ++ (2, 0.65) node[pos=.5]{$\hat{f}$};
            \draw[thick, rounded corners, fill=lightgray!30] (2.5, 3) rectangle ++ (2, 0.65) node[pos=.5]{$\hat{f}$};
            \draw[thick, rounded corners, fill=lightgray!30] (5, 3) rectangle ++ (2, 0.65) node[pos=.5]{$\hat{f}$};
            \draw[thick, rounded corners, fill=lightgray!30] (7.5, 3) rectangle ++ (2, 0.65) node[pos=.5]{$\hat{f}$};
            \draw[thick, rounded corners, fill=lightgray!30] (0, 1.5) rectangle ++ (2, 0.65) node[pos=.5]{$\hat{f}$};
            \draw[thick, rounded corners, fill=lightgray!30] (2.5, 1.5) rectangle ++ (2, 0.65) node[pos=.5]{$\hat{f}$};
            \draw[thick, rounded corners, fill=lightgray!30] (5, 1.5) rectangle ++ (2, 0.65) node[pos=.5]{$\hat{f}$};
            \draw[thick, rounded corners, fill=lightgray!30] (7.5, 1.5) rectangle ++ (2, 0.65) node[pos=.5]{$\hat{f}$};
            \draw[thick] (1, 2.6) -- (3, 2.6);
            \draw[->, thick, >=stealth] (3, 2.6) -- (3, 2.15);
            \draw[->, thick, >=stealth] (1, 3) -- (1, 2.15);
            \draw[thick] (1.5, 2.4) -- (5.5, 2.4);
            \draw[->, thick, >=stealth] (1.5, 2.4) -- (1.5, 2.15);
            \draw[->, thick, >=stealth] (5.5, 2.4) -- (5.5, 2.15);
            \draw[->, thick, >=stealth] (3.5, 3) -- (3.5, 2.15);
            \draw[thick] (4, 2.6) -- (8, 2.6);
            \draw[->, thick, >=stealth] (4, 2.6) -- (4, 2.15);
            \draw[->, thick, >=stealth] (8, 2.6) -- (8, 2.15);
            \draw[->, thick, >=stealth] (6, 3) -- (6, 2.15);
            \draw[thick] (6.5, 2.4) -- (8.5, 2.4);
            \draw[->, thick, >=stealth] (6.5, 2.4) -- (6.5, 2.15);
            \draw[->, thick, >=stealth] (8.5, 3) -- (8.5, 2.15);
            \draw[->, thick, >=stealth] (1, 1.5) -- (1, 0.8);
            \draw[->, thick, >=stealth] (3.5, 1.5) -- (3.5, 0.8);
            \draw[->, thick, >=stealth] (6, 1.5) -- (6, 0.8);
            \draw[->, thick, >=stealth] (8.5, 1.5) -- (8.5, 0.8);
            \node[] at (1, 0.5) {$\ldots$};
            \node[] at (3.5, 0.5) {$\ldots$};
            \node[] at (6, 0.5) {$\ldots$};
            \node[] at (8.5, 0.5) {$\ldots$};
        \end{tikzpicture}
        \caption{Embedding the space-time of $1$-dimensional first neighbors CA $A$ launched on an ultimately periodic configuration into an ultimately periodic configuration of a $2$-dimensional CA $B$.}
        \label{fig:logical-universality}
    \end{figure}
    
\end{proof}

\section{Appendix: Gates Implementation}\label{sec:app-gates}

We can combine our custom delay-insensitive gates to create asynchronous \textsf{NAND} (\Cref{fig:nand}) and fan-out (\Cref{fig:fan-out}) using dual-rail encoding. To verify the correctness of our construction we have implemented these gates in \texttt{Golly} \cite{golly}, a cellular automaton simulator and observed their space-time evolution. We use $\raisebox{-3.5pt}{\CustomSmall{w}}$ for $0$, $\raisebox{-3.5pt}{\CustomSmall{g}}$ for $1$ and $\raisebox{-3.5pt}{\CustomSmall{b}}$ for $2$ for a more clear representation. 

\begin{figure}[H]
    \centering
    \includegraphics[width=0.3\linewidth]{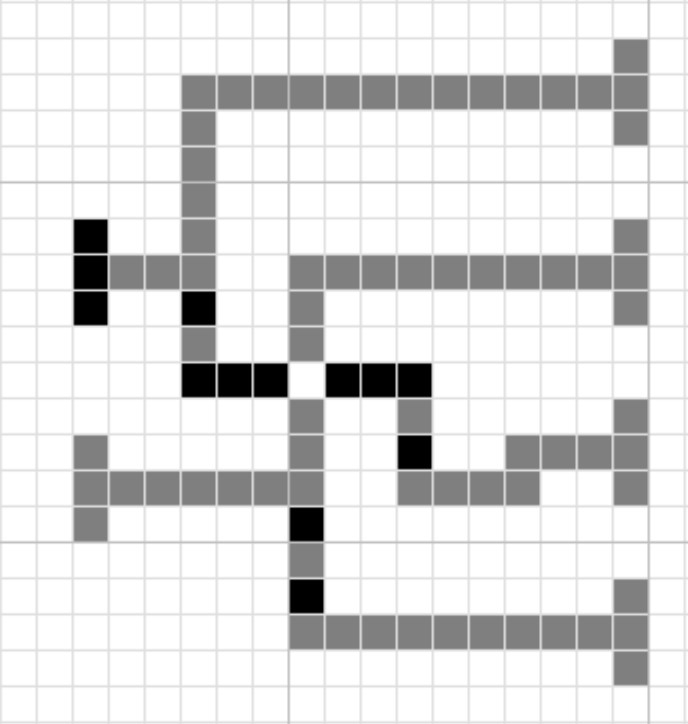}
    \caption{Implementation of fan-out pattern for the ACA $\mathbb{X}$ using dual-rail encoding. The input is set to true (the first ``tube'' on the left is colored black making its right neighbor active).}
    \label{fig:fan-out}
\end{figure}

\begin{figure}[H]
    \centering
    \includegraphics[width=0.8\linewidth]{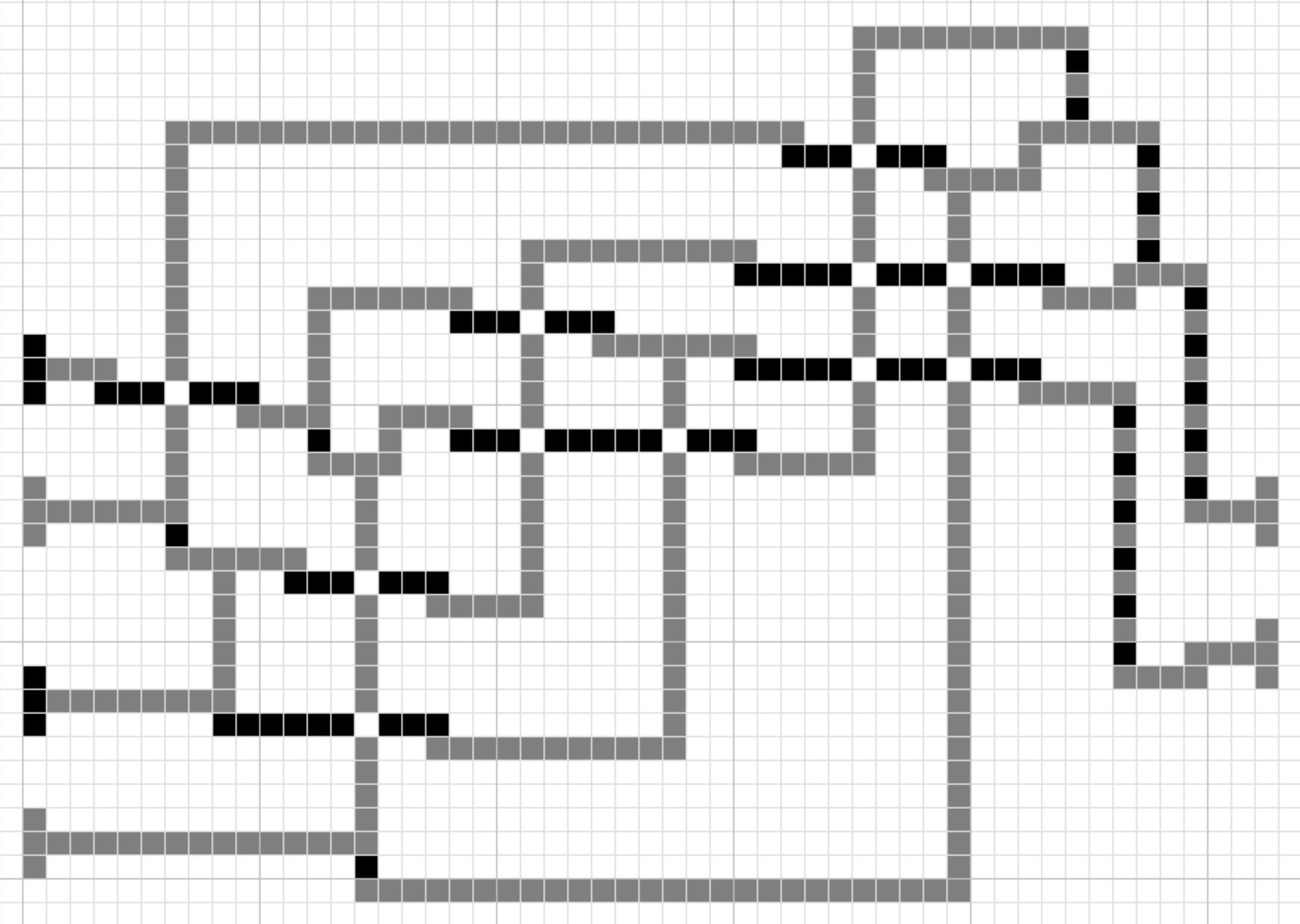}
    \caption{Implementation of \textsf{NAND} gate for the ACA $\mathbb{X}$ using dual-rail encoding. In this case, both inputs are initialized to true (the first and third ``tubes'' on the left are colored black making their right neighbors active).}
    \label{fig:nand}
\end{figure}

\fi

\end{document}